%% file: CADE_2021.tex
\def \inAppendix{}
 
\newcommand{\ShortLongVersion}[2]{#2}
\newcommand{\putInAppendix}[2]{#2}
\newcommand{\optproof}[2]{\proof{#2}}

\ShortLongVersion{
  \documentclass[envcountsame]{llncs}
}{
  \documentclass{article}
}

\usepackage[final]{commenting}
\usepackage{pslatex}
\usepackage{stmaryrd} 
\usepackage{amsfonts}
\usepackage{amssymb}
\usepackage{amsmath}
\usepackage{cmll}
\ShortLongVersion{}{
\usepackage{amsthm}}
\usepackage{scalerel}
\usepackage{xspace}
\usepackage{bm}
\usepackage{mfirstuc}
\usepackage[colorlinks]{hyperref}
\usepackage{paralist}
\usepackage{url}
\usepackage{tablefootnote}

\declareauthor{ri}{Radu}{blue}
\declareauthor{np}{Nicolas}{red}
\declareauthor{me}{Mnacho}{brown!70!black}

\newcommand{\astore}{{\frak s}}
\newcommand{\aheap}{{\frak h}}

\newcommand{\asid}{{\cal R}}
\newcommand{\rsid}{{\widehat{\cal R}}}
\newcommand{\isdef}{\stackrel{\scriptscriptstyle{\mathsf{def}}}{=}}

\newcommand{\interv}[2]{\llbracket #1,#2 \rrbracket}
\newcommand{\len}[1]{|#1|}
\newcommand{\size}[1]{|#1|}
\newcommand{\width}[1]{\mathrm{w}(#1)}
\newcommand{\lang}[1]{\mathcal{L}(#1)}
\newcommand{\nat}{\mathbb{N}}
\newcommand{\rank}{\kappa}
\newcommand{\id}{{\mathit id}}

\newcommand{\dom}[1]{\mathrm{dom}(#1)}
\newcommand{\img}[1]{\mathrm{rng}(#1)}
\newcommand{\imgh}[1]{\mathrm{ref}(#1)}
\newcommand{\locs}[1]{\mathrm{loc}(#1)}
\newcommand{\iseq}{\approx}

\newcommand{\vars}{\mathsf{V}}
\newcommand{\fv}[1]{\mathrm{fv}(#1)}
\newcommand{\card}[1]{|#1|}

\newcommand{\apl}[2]{#1(#2)}
\newcommand{\astorep}{\astore_e}

\ShortLongVersion{
\spnewtheorem{assumption}[theorem]{Assumption}{\bfseries}{\itshape}}{
\newtheorem{theorem}{Theorem}
\newtheorem{definition}[theorem]{Definition}
\newtheorem{lemma}[theorem]{Lemma}
\newtheorem{proposition}[theorem]{Proposition}

\newtheorem{assumption}[theorem]{Assumption}

\newtheorem{example}[theorem]{Example}
}

\newcommand{\pb}[3]{#1 \vdash_{#3} #2}
\newcommand{\aprob}{\mathfrak{P}}
\newcommand{\dunion}{\uplus}
\newcommand{\inv}[1]{{#1}^{-1}}

\newcommand{\genbod}{\pi}
\renewcommand{\Asterisk}{\scalebox{1.8}{\text{$*$}}}
\newcommand{\twoexptime}{$\mathsf{2\text{EXPTIME}}$}
\newcommand{\vargf}{{\vars}_{\fvargs}}
\newcommand{\vargs}[1]{\vargf(#1)}
\newcommand{\rpositional}{$\asid$-positional\xspace}
\newcommand{\trunc}[1]{\mathrm{trunc}(#1)}
\newcommand{\decor}[1]{D(#1)}

\begin{document}
%%%%%%%%%%%%%%%%%%%%%%%%%%%%%%%%%%%%%%%%%%%%%%%%%%%%%%%%%%%%%%%%%%%%%%%%%%%%%

% \title{On the Complexity of the Entailment Problem in Separation Logic with Recursive Definitions?
% or Safe Entailments in Separation Logic with Recursive Definitions?}

\title{Unifying Decidable Entailments in Separation Logic with Inductive Definitions}

\ShortLongVersion{
%\titlerunning{Safe Entailments in Separation Logic with Recursive Definitions}
\author{Mnacho Echenim\inst{1}, Radu Iosif\inst{2} and Nicolas Peltier\inst{1}}

\institute{Univ. Grenoble Alpes, CNRS, LIG, F-38000 Grenoble France
\and
Univ. Grenoble Alpes, CNRS, VERIMAG, F-38000 Grenoble France}
\authorrunning{M. Echenim et al.}
}
{
\author{Mnacho Echenim \and Radu Iosif \and Nicolas Peltier\\ 
Univ. Grenoble Alpes, CNRS, LIG/Verimag, F-38000 Grenoble France
}

}

\maketitle

\newcommand{\fvprofile}{fv-profile\xspace}
\newcommand{\fvargsof}[2]{\fvargs^{#1}_{#2}}

\newcommand{\safe}{safe\xspace}
\newcommand{\fvargs}{\lambda}
\newcommand{\dependson}[1]{\succeq_{#1}}
\newcommand{\arity}[1]{\mathit{ar}(#1)}
\newcommand{\emp}{\mathit{emp}}

\newcommand{\Loc}{{\cal L}}
\newcommand{\bots}{\Bot}
\newcommand{\preds}[1]{{\cal P}(#1)}
\newcommand{\allpreds}{\mathsf{P}}

\newcommand{\bigast}{\scaleobj{1.5}{\ast}}

\newcommand{\set}[1]{\left\{#1\right\}}
 \newcommand{\setof}[2]{\left\{#1\,\middle|\,#2\right\}}
 \newcommand{\tuple}[1]{\langle#1\rangle}
 
 \newcommand{\mvec}[1]{\pmb{#1}}

\newcommand{\progressing}{$\lambda$-progressing\xspace}
\newcommand{\connected}{$\lambda$-connected\xspace}
\newcommand{\C}{$\lambda$-C\xspace}
\newcommand{\restricted}{$\lambda$-restricted\xspace}
\newcommand{\R}{$\lambda$-R\xspace}
\newcommand{\vrestricted}{$\fvar$-restricted\xspace}
\newcommand{\genrest}{\rho}

\newcommand{\subst}[2]{\set{ (#1,#2) }}
\newcommand{\substmult}[4]{ \set{ \tuple{#1,#2},\dots,\tuple{#3,#4} } }
\newcommand{\substinterv}[4]{ \set{ \tuple{#1,#2} \mid #3 \in #4 }}
\newcommand{\bigO}{\mathcal{O}}

\newcommand{\maxv}{\mu}

\newcommand{\est}[1]{#1_E}
\newcommand{\fvar}{w}
\newcommand{\nfv}{\nu}
\newcommand{\leftpred}{{\cal P}_{l}}
\newcommand{\rightpred}{{\cal P}_{r}}
\newcommand{\botp}{\underline{\boldsymbol{\bot}}}

\newcommand{\expands}[1]{\triangleright_{#1}}
\newcommand{\gconnection}{connection\xspace}
\newcommand{\gcon}[2]{\mathrm{C}_{#2}(#1)}

\newcommand{\bpart}[1]{\mathrm{aux}(#1)}
\newcommand{\mpart}[1]{\mathrm{main}(#1)}

%% \begin{abstract}
%% We show that the entailment problem $\varphi \models \psi$ in
%% Separation Logic is decidable for separated conjunctions of atoms
%% $\varphi$ and $\psi$, that contain predicate symbols whose
%% interpretation is given inductively by a set of recursive rules. The
%% proof is based on a reduction to a class of entailment problems shown
%% to be decidable in \cite{IosifRogalewiczSimacek13}. In contrast with
%% the works of
%% \cite{IosifRogalewiczSimacek13,KatelaanMathejaZuleger19,PZ20}, the
%% considered inductive rules may introduce memory locations without
%% allocating them, which strongly extends the class of structures that
%% can be constructed. Moreover, the result is also much \comment[me]{strictly?} more general than the one
%% given in \cite{EIP21a}, because the conditions on the inductive rules
%% corresponding to the left-hand side of the considered entailment are
%% strongly relaxed: it is only assumed that the rules are progressing,
%% i.e.\ that they allocate exactly one memory location.
%% \end{abstract}

\begin{abstract}
\comment[ri]{rewritten, old version in comments} The entailment
problem $\varphi \models \psi$ in Separation Logic
\cite{IshtiaqOHearn01,Reynolds02}, between separated conjunctions of
equational ($x \iseq y$ and $x \not\iseq y$), spatial ($x \mapsto
(y_1,\ldots,y_\rank)$) and predicate ($p(x_1,\ldots,x_n)$) atoms,
interpreted by a finite set of inductive rules, is undecidable in
general. Certain restrictions on the set of inductive definitions lead
to decidable classes of entailment problems. Currently, there are two
such decidable classes, based on two restrictions, called
\emph{establishment}
\cite{IosifRogalewiczSimacek13,KatelaanMathejaZuleger19,PZ20} and
\emph{restrictedness} \cite{EIP21a}, respectively.  Both classes are
shown to be in \twoexptime\ by the independent proofs from \cite{PZ20}
and \cite{EIP21a}, respectively, and a many-one reduction of
established to restricted entailment problems has been given
\cite{EIP21a}. In this paper, we strictly generalize the restricted
class, by distinguishing the conditions that apply only to the left-
($\varphi$) and the right- ($\psi$) hand side of entailments,
respectively. We provide a many-one reduction of this generalized
class, called \emph{safe}, to the established class. Together with the
reduction of established to restricted entailment problems, this new
reduction closes the loop and shows that the three classes of
entailment problems (respectively established, restricted and safe)
form a single, unified, \twoexptime-complete class.
\end{abstract}

\section{\capitalisewords{Introduction}}

%% \comment[np]{I changed the title because it was too close from that
%% of the previous paper (this is very risky, the main criticism for
%% the paper is that it is incremental wrt CSL paper). All ideas are
%% welcome :) }

%% \comment[me]{I like the second title :)}

%% \comment[np]{To improve/extend. We have to show that the results
%% are significant especially wrt CSL paper.  I hope this is
%% convincing enough :) }

\comment[ri]{changes below}

Separation Logic \cite{IshtiaqOHearn01,Reynolds02} (SL) was primarily
introduced for writing concise Hoare logic proofs of programs that
handle pointer-linked recursive data structures (lists, trees,
etc). Over time, SL has evolved into a powerful logical framework,
that constitutes the basis of several industrial-scale static program
analyzers
\cite{DBLP:conf/nfm/CalcagnoDDGHLOP15,DBLP:conf/cav/BerdineCI11,DBLP:conf/cav/DudkaPV11},
that perform scalable compositional analyses, based on the principle
of \emph{local reasoning}: describing the behavior of a program
statement  with respect only to the small (local) set of memory
locations that are changed by that statement, with no concern for the
rest of the program's state.

Given a set of memory locations (e.g., addresses), SL formul{\ae}
describe {\em heaps}, that are finite partial functions mapping
finitely many locations to records of locations. A location $\ell$ is
{\em allocated} if it occurs in the domain of the heap. An atom
$x \mapsto (y_1,\dots,y_\rank)$ states that there is only one
allocated location, associated with $x$, that moreover refers to the
tuple of locations associated with $(y_1,\dots,y_\rank)$,
respectively. The {\em separating conjunction} $\phi * \psi$ states
that the heap can split into two parts, with disjoint domains, that
make $\phi$ and $\psi$ true, respectively. The separating conjunction
is instrumental in supporting local reasoning, because the
disjointness between the (domains of the) models of its arguments
ensures that no update of one heap can actually affect the other.

Reasoning about recursive data structures of unbounded sizes (lists,
trees, etc.) is possible via the use of predicate symbols, whose
interpretation is specified by a user-provided \emph{set of inductive
  definitions} (SID) of the form $p(x_1,\ldots,x_n) \Leftarrow \pi$,
where $p$ is a predicate symbol of arity $n$ and the free variables of
the formula $\pi$ are among the parameters $x_1,\ldots,x_n$ of the
rule. Here the separating conjunction ensures that each unfolding of
the rules, which substitute some predicate atom $p(y_1,\ldots,y_n)$ by
a formula $\pi[x_1/y_1,\ldots,x_n/y_n]$, corresponds to a way of
building the recursive data structure. For instance, a list is either
empty, in which case its head equals its tail pointer, or is built by
first allocating the head, followed by all elements up to but not
including the tail, as stated by the inductive definitions
\(\mathsf{ls}(x,y) \Leftarrow x \iseq y\) and \(\mathsf{ls}(x,y)
\Leftarrow \exists z ~.~ x \mapsto (z) * \mathsf{ls}(z,y)\).

An important problem in program verification, arising during the
construction of Hoare-style correctness proofs of programs, is the
discharge of verification conditions of the form $\phi \models \psi$,
where $\phi$ and $\psi$ are SL formul{\ae}, asking whether every model
of $\phi$ is also a model of $\psi$. These problems,
called \emph{entailments}, are, in general, undecidable in the
presence of inductively defined predicates
\cite{DBLP:conf/atva/IosifRV14,AntonopoulosGorogiannisHaaseKanovichOuaknine14}.

A first decidable class of entailments, described in
\cite{IosifRogalewiczSimacek13}, involves three restrictions on the
SID rules: \emph{progress}, \emph{connectivity} and
\emph{establishment}. Intuitively, the progress (P) condition states
that every rule allocates exactly one location, the connectivity (C)
condition states that the set of allocated locations has a tree-shaped
structure, and the establishment (E) condition states that every
existentially quantified variable from a rule defining a predicate is
(eventually) allocated in every unfolding of that predicate.  A
\twoexptime\ algorithm was proposed for testing the validity of PCE
entailments \cite{KatelaanMathejaZuleger19,PZ20} and a matching
\twoexptime-hardness lower bound was provided shortly after
\cite{DBLP:conf/lpar/EchenimIP20}.

Later work relaxes the establishment condition, necessary for
decidability \cite{echenim:hal-02951630}, by proving that the
entailment problem is still in \twoexptime\ if the establishment
condition is replaced by the {\em restrictedness} (R) condition, which
requires that every disequality ($x \not\iseq y$) involves at least
one free variable from the left-hand side of the entailment,
propagated through the unfoldings of the inductive system
\cite{EIP21a}. Interestingly, the rules of a progressive, connected
and restricted (PCR) entailment may generate data structures with
``dangling'' (i.e.\ existentially quantified but not allocated)
pointers, which was not possible with PCE entailments.

%% One interesting feature of this class of entailment is that it allows
%% one to handle rules generating data structures with ``dangling''
%% edges, for instance rules
%% \comment[ri]{removed: lists} with pending elements.  We also slightly
%% generalized the connectivity condition, by allowing forests (rooted on
%% free variables) instead of trees.

In this paper, we generalize PCR entailments further, by showing that
the connectivity and restrictedness conditions are needed only on the
right-hand side of the entailment, whereas the only condition required
on the left-hand side is progress. Although the class of data
structures that can be described is much larger, we show that this new
class of entailments, called \emph{safe}, is also \twoexptime-complete, by
a many-one reduction of the validity of safe entailments to the
validity of PCE entailments. A second contribution of the paper is the
cross-certification of the two independent proofs of the
\twoexptime\ upper bounds, for the PCE
\cite{DBLP:conf/lpar/EchenimIP20,PZ20,EIP21a} and PCR \cite{EIP21a}
classes of entailments, respectively, by closing the loop. Namely, the
reduction given in this paper enables the translation of any of the
three entailment problems into an equivalent problem in any other class,
while preserving the \twoexptime\ upper bound. This is because all the
reductions are polynomial in the overall size of the SID and
singly-exponential in the maximum size of the rules in the SID.

Due to space restrictions, some of the proofs are shifted to the
Appendix.

\section{\capitalisewords{Definitions}}

For a (partial) function $f : A \rightarrow B$, we denote by $\dom{f}$
and $\img{f}$ its domain and range,
respectively. \ShortLongVersion{}{A function $f$ is {\em finite} if
  $\card{\dom{f}} < \infty$, where $\card{S}$ denotes the cardinality
  of the set $S$. The subset $\set{k,k+1,\ldots,\ell}$ of the set
  $\nat$ of natural numbers is denoted as $\interv{k}{\ell}$; note
  that $\interv{k}{\ell} = \emptyset$ whenever $\ell < k$.} For a
relation $R \subseteq A \times A$, we denote by $R^*$ the reflexive
and transitive closure of $R$\ShortLongVersion{}{, i.e.\ $R^* \isdef
  \{(x_1,x_n) \mid n \geq 1,~ \forall i \in \interv{1}{n-1} ~.~ (x_i,
  x_{i+1}) \in R\}$}.

Let $\rank$ be a fixed natural number throughout this paper and let
$\allpreds$ be a countably infinite set of {\em predicate
  symbols}. Each predicate symbol $p \in \allpreds$ is associated a
unique arity, denoted $\arity{p}$.  Let $\vars$ be a countably
infinite set of {\em variables}. For technical convenience, we also
consider a special constant $\bot$, which will be used to denote
``empty'' record fields.  Formul{\ae} are built inductively, according
to the following syntax:
\[\phi := x \not \iseq x' \mid x \iseq x' \mid x \mapsto (y_1,\dots,y_\rank) \mid
p(x_1,\dots,x_n) \mid \phi_1
* \phi_2 \mid \phi_1 \vee \phi_2 \mid \exists x ~.~ \phi_1\] where
$p\in \allpreds$ is a predicate symbol of arity $n = \arity{p}$,
$x,x',x_1,\dots,x_n \in \vars$ are variables and
$y_1,\dots,y_\rank \in \vars \cup \{ \bot \}$ are \emph{terms}, i.e.\
either variables or $\bot$.

The set of variables freely occurring in a formula $\phi$ is denoted
by $\fv{\phi}$, we assume by $\alpha$-equivalence that the same
variable cannot occur both free and bound in the same formula $\phi$,
and that distinct quantifiers bind distinct variables.  The {\em size}
$\size{\phi}$ of a formula $\phi$ is the number of occurrences of
symbols in $\phi$.  A formula $x \iseq x'$ or $x \not\iseq x'$ is an
\emph{equational atom}, $x \mapsto (y_1, \ldots, y_\rank)$ is a
\emph{points-to atom}, whereas $p(x_1, \ldots, x_n)$ is a
\emph{predicate atom}.  \comment[ri]{added} Note that $\bot$ cannot
occur in an equational or in a predicate atom.  A formula is
\emph{predicate-less} if no predicate atom occurs in it. A
\emph{symbolic heap} is a formula \comment[np]{deleted: containing no disjunctions, i.e., (because strictly speaking this is not equivalent, equivalence is explained latter)} of
the form $\exists \mvec{x}~.~ \Asterisk_{j=1}^{m} \alpha_i$, where
each $\alpha_i$ is an atom and $\mvec{x}$ is a possibly empty vector of variables.

\begin{definition}\label{def:alloc}
  A variable $x$ is \emph{allocated by a symbolic heap} $\phi$ iff
  $\phi$ contains a sequence of equalities $x_1 \iseq x_2 \iseq \ldots
  \iseq x_{n-1} \iseq x_n$, for $n \geq 1$, such that $x = x_1$ and
  $x_n \mapsto (y_1, \ldots, y_\rank)$ occurs in $\phi$, for some
  variables $x_1, \ldots, x_n$ and some terms $y_1, \ldots, y_\rank
  \in \vars \cup \{\bot\}$.
\end{definition}
A {\em substitution} is a partial function mapping variables to
variables. If $\sigma$ is a substitution and $\phi$ is a formula, a
variable or a tuple, then $\phi\sigma$ denotes the formula, the
variable or the tuple obtained from $\phi$ by replacing every free
occurrence of a variable $x \in \dom{\sigma}$ by $\sigma(x)$,
respectively. We denote by $\substinterv{x_i}{y_i}{i}{\interv{1}{n}}$
the substitution with domain $\{ x_1,\dots,x_n\}$ that maps $x_i$ to
$y_i$, for each $i \in
\interv{1}{n}$.

\comment[ri]{moved these 2 paragraphs up here and changed: system
$\rightarrow$ set in SID}

A {\em set of inductive definitions} (SID) is a finite set $\asid$ of
rules of the form $p(x_1,\dots,x_n) \Leftarrow \genbod$, where $p \in
\allpreds$, $n = \arity{p}$, $x_1,\dots,x_n$ are pairwise distinct
variables and $\genbod$ is a quantifier-free symbolic heap. The
predicate atom $p(x_1,\ldots,x_n)$ is the \emph{head} of the rule and
$\asid(p)$ denotes the subset of $\asid$ consisting of rules with head
$p(x_1,\ldots,x_n)$ (the choice of $x_1, \ldots, x_n$ is not
important). The variables in $\fv{\genbod} \setminus
\{x_1,\dots,x_n\}$ are called the \emph{existential variables of the
  rule}. Note that, by definition, these variables are not explicitly
quantified inside $\genbod$ and that $\genbod$ is quantifier-free. For
simplicity, we denote by $p(x_1, \ldots,
x_n) \Leftarrow_\asid \genbod$ the fact that the rule $p(x_1, \ldots,
x_n) \Leftarrow \genbod$ belongs to $\asid$. The
\emph{size} of $\asid$ is defined as $\size{\asid} \isdef \sum_{p(x_1,
  \ldots, x_n) \Leftarrow_\asid \pi} \size{\pi} + n$ and
its \emph{width} as $\width{\asid} \isdef
\max_{p(x_1, \ldots, x_n) \Leftarrow_\asid \pi} \size{\pi} + n$.

We write $p \dependson{\asid} q$, $p, q \in \allpreds$ iff $\asid$
contains a rule of the form $p(x_1, \ldots, x_n) \Leftarrow
\genbod$, and $q$ occurs in $\genbod$. We say that $p$ \emph{depends on}
$q$ if $p \dependson{\asid}^* q$.  For a formula $\phi$, we denote by
$\preds{\phi}$ the set of predicate symbols $q$, such that $p \dependson{\asid}^* q$ for some predicate $p$ occurring in $\phi$. \comment[np]{instead of ``such that $p$ occurs
in $\phi$ and $p \dependson{\asid}^* q$.''}

Given formul{\ae} $\phi$ and $\psi$, we write $\phi \Leftarrow_{\asid}
\psi$ if $\psi$ is obtained from $\phi$ by replacing an atom
$p(u_1,\dots,u_n)$ by $\genbod\substmult{x_1}{u_1}{x_n}{u_n}$, where
$\asid$ contains a rule $p(x_1,\ldots,x_n) \Leftarrow \genbod$. 
We assume, by a renaming of
existential variables, that $(\fv{\pi} \setminus \{
x_1,\dots,x_n\}) \cap \fv{\phi} = \emptyset$. %, for any store
%$\astore$ and any rule $p(x_1,\dots,x_n) \Leftarrow_{\asid} \pi$.}
\comment[np]{instead of (stores are not yet defined, moreover, this is not really $\alpha$-renaming): We assume, by an $\alpha$-renaming of
existential variables, if necessary, that $(\fv{\pi} \setminus \{
x_1,\dots,x_n\}) \cap \dom{\astore} = \emptyset$, for any store
$\astore$ and any rule $p(x_1,\dots,x_n) \Leftarrow_{\asid} \pi$.}
We
call $\psi$ an {\em unfolding} of $\phi$ iff $\phi
\Leftarrow_{\asid}^* \psi$.

\comment[ri]{added back in the paper}
\begin{proposition}\label{lemma:unfolding}
  Every unfolding of a symbolic heap is again a symbolic heap.
\end{proposition}
\begin{proof}
By induction on the length of the unfolding sequence. \qed
\end{proof}

We now define the semantics of SL. Let $\Loc$ be a countably infinite
set of {\em locations} containing, in particular, a special location
$\bots$. A {\em structure} is a pair $(\astore,\aheap)$,
where: \begin{compactitem}
\item $\astore$ is a partial function from $\vars \cup \{ \bot \}$ to
  $\Loc$, called a {\em store}, such that $\bot\in \dom{\astore}$ and
  $\astore(x) = \bots \iff x = \bot$, \comment[ri]{added}  for all
  $x\in\vars \cup \{ \bot \}$, \comment[np]{instead of $x\in \vars$} and
\item $\aheap : \Loc \rightarrow \Loc^\rank$ is a finite partial function,
      such that $\bots\not \in \dom{\aheap}$.           
\end{compactitem}
If $x_1,\dots,x_n$ are pairwise distinct variables and
$\ell_1,\dots,\ell_n \in \Loc$ are locations, we denote by
$\astore[x_i \leftarrow \ell_i \mid 1 \leq i \leq n]$ the store
$\astore'$ defined by $\dom{\astore'} = \dom{\astore} \cup \set{x_1,
  \ldots, x_n}$, $\astore'(y) = \ell_i$ if $y = x_i$ for some $i \in
\interv{1}{n}$, and $\astore'(y) = \astore(x)$ otherwise. If
$x_1,\dots,x_n \not\in \dom{\astore}$, then the store $\astore'$ is
called an {\em extension} of $\astore$ to $\{x_1,\dots,x_n\}$.

Given a heap $\aheap$, we define $\imgh{\aheap} \isdef \bigcup_{l\in
  \dom{\aheap}}\{ \ell_i \mid \aheap(\ell) =
(\ell_1,\dots,\ell_\rank), i \in \interv{1}{\rank}\}$ and
$\locs{\aheap} \isdef \dom{\aheap} \cup \imgh{\aheap}$.  Two heaps
$\aheap_1$ and $\aheap_2$ are {\em disjoint} iff $\dom{\aheap_1} \cap
\dom{\aheap_2} = \emptyset$, in which case $\aheap_1 \dunion \aheap_2$
denotes the union of $\aheap_1$ and $\aheap_2$, undefined whenever
$\aheap_1$ and $\aheap_2$ are not disjoint.

Given an SID $\asid$, $(\astore,\aheap) \models_{\asid} \phi$ is the
least relation between structures and formul{\ae} such that whenever
$(\astore,\aheap) \models_{\asid} \phi$, we have $\fv{\phi} \subseteq
\dom{\astore}$ and the following hold: \comment[np]{change in the last item below, to handle pb with variable collisions}
\[\begin{array}{rcll}
(\astore,\aheap) & \models_\asid & x \iseq x' & \text{ if $\dom{\aheap} = \emptyset$ and $\astore(x) = \astore(x')$} \\
(\astore,\aheap) & \models_\asid & x \not\iseq x' & \text{ if $\dom{\aheap} = \emptyset$ and $\astore(x) \neq \astore(x')$} \\
(\astore,\aheap) & \models_\asid & x \mapsto (y_1, \ldots, y_\rank) & \text{ if $\dom{\aheap} = \{\astore(x)\}$ and 
$\aheap(\astore(x)) = \tuple{\astore(y_1), \ldots, \astore(y_\rank)}$} \\
(\astore,\aheap) & \models_\asid & \phi_1 * \phi_2 & \text{ if there exist disjoint heaps $\aheap_1$ and $\aheap_2$ such that} \\
&&& \text{ $\aheap = \aheap_1 \dunion \aheap_2$ and $(\astore,\aheap_i) \models_\asid \phi_i$, for both $i=1,2$} \\
(\astore,\aheap) & \models_\asid & \phi_1 \vee \phi_2 & \text{ if $(\astore,\aheap) \models_\asid \phi_i$, for some $i = 1,2$} \\
(\astore,\aheap) & \models_\asid & \exists x~.~ \phi & \text{ if there exists $\ell\in \Loc$ such that $(\astore[x \leftarrow \ell],\aheap) \models\phi$} \\
(\astore,\aheap) & \models_\asid & p(x_1,\dots,x_n) & \text{ if $p(x_1,\dots,x_n) \Leftarrow_{\asid} \phi$, and there exists a store $\astore_e$} \\
&&& \text{ coinciding with $\astore$ on $\{ x_1,\dots,x_n\}$, such that $(\astore_e,\aheap) \models \phi$} 
%(\astore,\aheap) & \models_\asid & p(x_1,\dots,x_n) & \text{ if $p(x_1,\dots,x_n) \Leftarrow_{\asid} \phi$, and there exists an extension} \\
%&&& \text{ $\astore_e$ of $\astore$ to $\fv{\phi} \setminus \dom{\astore}$, such that $(\astore_e,\aheap) \models \phi$} 
\end{array}\]
%% \begin{itemize}
%% \item{If $\aheap = \emptyset$, $x,y\in \dom{\astore}$ and $\astore(x) = \astore(y)$ then $(\astore,\aheap) \models x \iseq y$}
%% \item{If $\aheap = \emptyset$, $x,y\in \dom{\astore}$ and $\astore(x) \not =  \astore(y)$ then $(\astore,\aheap) \models x \not\iseq y$}
%% \item{If $\aheap_1$ and $\aheap_2$ are disjoint heaps  such that 
%%  $(\astore,\aheap_i) \models \phi_i$ for $i = 1,2$, then 
%% $(\astore,\aheap_1 \dunion \aheap_2) \models \phi_1 * \phi_2$.}
%% \item{If $(\astore,\aheap_i) \models \phi_i$ for some $i = 1,2$, then 
%% $(\astore,\aheap) \models \phi_1 \vee \phi_2$.}
%% \item{If there exists $\ell \in \Loc$ such that $(\astore[x \leftarrow \ell],\aheap) \models \phi$ 
%% then 
%% $(\astore,\aheap) \models \exists x.~\phi$.}
%% \item{If $p(x_1,\dots,x_n) \Leftarrow_{\asid} \phi$, $\set{x_1,\dots,x_n} \subseteq \dom{\astore}$ and 
%% there exists an extension $\astore_e$ of $\astore$ to $\fv{\phi} \setminus \dom{\astore}$ such that
%% $(\astore_e,\aheap) \models \phi$,  then
%% $(\astore,\aheap) \models p(x_1,\dots,x_n)$.} 
%% \end{itemize}
%\comment[np]{added:} We assume, by a renaming of
%existential variables, that $\fv{\phi} \cap \dom{\astore} \subseteq \{ x_1,\dots,x_n \}$.

Given formul{\ae} $\phi$ and $\psi$, we write $\phi \models_{\asid}
\psi$ whenever $(\astore,\aheap) \models_{\asid} \phi \Rightarrow
(\astore,\aheap) \models_{\asid} \psi$, for all structures
$(\astore,\aheap)$ and $\phi \equiv_{\asid} \psi$ for ($\phi
\models_\asid \psi$ and $\psi \models_\asid \phi$). We omit the
subscript $\asid$ whenever these relations hold for any SID.  It is
easy to check that, for all formul{\ae} $\phi_1,\phi_2,\psi$, it is the
case that $(\phi_1 \vee \phi_2) * \psi \equiv (\phi_1 * \psi) \vee
(\phi_2 * \psi)$ and $(\exists x . \phi_1) * \phi_2 \equiv \exists x ~.~ \phi_1
* \phi_2$. Consequently, each formula can be transformed into an
equivalent finite disjunction of symbolic heaps.

\begin{definition}\label{def:entailment}
An {\em entailment problem} is a triple $\aprob \isdef
\pb{\phi}{\psi}{\asid}$, where $\phi$ is a quantifier-free formula,
$\psi$ is a formula and $\asid$ is an SID. The problem $\aprob$ is
\emph{valid} iff $\phi \models_{\asid} \psi$. The \emph{size} of the
problem $\aprob$ is defined as $\size{\aprob} \isdef
\size{\phi}+\size{\psi}+\size{\asid}$ and its \emph{width} is defined
as $\width{\aprob} \isdef \max(\size{\phi}, \size{\psi},
\width{\asid})$.
\end{definition}
\noindent
Note that considering $\phi$ to be quantifier-free loses no
generality, because $\exists x.\phi \models_{\asid} \psi \iff \phi
\models_{\asid} \psi$.

\section{Decidable Entailment Problems}
\label{sec:decidable-entailments}

The class of general entailment problems is undecidable, see
Theorem \ref{thm:right-disc-undec} below for a refinement of the
initial undecidability proofs
\cite{DBLP:conf/atva/IosifRV14,AntonopoulosGorogiannisHaaseKanovichOuaknine14}. A first
attempt to define a natural decidable class of entailment problems is
described in \cite{IosifRogalewiczSimacek13} and involves three
restrictions on the SID rules, formally defined below:

\comment[ri]{abbreviations: P for progressing, C for connected, E for established}
\begin{definition}\label{def:pce}
  A rule $p(x_1,\dots,x_n) \Leftarrow \pi$ is: \begin{compactenum}
  \item\label{it:progressing} \emph{progressing} (P) iff $\pi = x_1
    \mapsto (y_1,\dots,y_\rank) * \genrest$ and $\genrest$ contains no
    points-to atoms,
  \item\label{it:connected} \emph{connected} (C) iff it is progressing,
    $\pi = x_1 \mapsto (y_1,\dots,y_\rank) * \genrest$ and every
    predicate atom in $\genrest$ is of the form $q(y_i,\mvec{u})$, for
    some $i \in \interv{1}{\rank}$,
  \item\label{it:established} \emph{established} (E) iff every
    existential variable $x \in \fv{\pi} \setminus \{x_1, \ldots,
    x_n\}$ is allocated by every predicate-less unfolding\footnote{As
      stated in Proposition \ref{lemma:unfolding}, $\phi$ is a
      symbolic heap.} $\pi \Leftarrow_\asid^* \phi$.
  \end{compactenum}
  An SID $\asid$ is \emph{P} (resp. \emph{C}, \emph{E}) \emph{for a
    formula} $\phi$ iff every rule in $\bigcup_{p \in
    \preds{\phi}}\asid(p)$ is P (resp. C,E). An entailment problem
  $\pb{\phi}{\psi}{\asid}$ is \emph{left-} (resp. \emph{right-})
  \emph{P} (resp. \emph{C}, \emph{E}) iff $\asid$ is P (resp. C, E)
  for $\phi$ (resp. $\psi$).  An entailment problem is \emph{P}
  (resp. \emph{C}, \emph{E}) iff it is both left- and right-P
  (resp. C, E).
\end{definition}
The decidability of progressing, connected and left-established
entailment problems is an immediate consequence of the result of
\cite{IosifRogalewiczSimacek13}. Moreover, an analysis of the proof
\cite{IosifRogalewiczSimacek13} leads to an elementary recursive
complexity upper bound, which has been recently tighten down to
\twoexptime-complete \cite{PZ20,EIP21a,DBLP:conf/lpar/EchenimIP20}.
In the following, we refer to Table \ref{tab:entl-complexity} for a
recap of the complexity results for the entailment problem. The last
line is the main result of the paper and corresponds to the most
general (known) decidable class of entailment problems (Definition
\ref{def:safe}).

\begin{table}[h!]
\begin{center}
\caption{Decidability and Complexity Results for the Entailment
  Problem ($\checkmark$ means that the corresponding condition
  holds on the left- and right-hand side of the entailment)}
\label{tab:entl-complexity}
\begin{tabular}{|c|c|c|c|c|c|}
\hline
Reference & Progress & Connected & Established & Restricted & Complexity \\
\hline\hline
Theorem \ref{thm:pce-complexity} & $\checkmark$ & $\checkmark$ & left & - & 2EXP-co. \\
\hline
Theorem \ref{thm:right-disc-undec} & $\checkmark$ & left & $\checkmark$ & - & undec. \\
\hline
\cite[Theorem 6]{echenim:hal-02951630} & $\checkmark$ & $\checkmark$ & - & - & undec. \\
\hline
\cite[Theorem 32]{EIP21a} & $\checkmark$ & $\checkmark$ & - & $\checkmark$ & 2EXP-co. \\ 
\hline
 Theorem \ref{thm:safe-complexity} & $\checkmark$ & right & - & right & 2EXP-co. \\
\hline
\end{tabular}
\end{center}
\vspace*{-2\baselineskip}
\end{table}

%\begin{table}[h!]
%\begin{center}
%\caption{Decidability and Complexity Results for the Entailment
%  Problem --- $\checkmark$ (resp. $\lambda$) means the condition
%  (resp. $\lambda$-condition) holds both left and right.}
%\label{tab:entl-complexity}
%\begin{tabular}{|c|c|c|c|c|c|}
%\hline
%Reference & Progress & Connected & Established & Restricted & Complexity \\
%\hline\hline
%Theorem \ref{thm:pce-complexity} & $\checkmark$ & $\checkmark$ & left & - & 2EXP-co. \\
%\hline
%Theorem \ref{thm:right-disc-undec} & $\checkmark$ & left & $\checkmark$ & - & undec. \\
%\hline
%\cite[Theorem 6]{echenim:hal-02951630} & $\checkmark$ & $\checkmark$ & - & - & undec. \\
%\hline
%\cite[Theorem 32]{EIP21a} & $\checkmark$ & $\lambda$ & - & $\lambda$ & 2EXP-co. \\ 
%\hline
%Theorem \ref{thm:safe-complexity} & $\checkmark$ & $\lambda$-right & - & $\lambda$-right & 2EXP-co. \\
%\hline
%\end{tabular}
%\end{center}
%\end{table}

The following theorem is an easy consequence of previous results \cite{DBLP:conf/lpar/EchenimIP20}.

\putInAppendix{\section{Additional Material for Section \ref{sec:decidable-entailments}}}{}

\begin{theorem}\label{thm:pce-complexity}
The progressing, connected and left-established entailment problem is
\twoexptime-complete. \comment[np]{slight change (the algo cannot depend on the instance, of course} Moreover, there exists an algorithm that runs in time
$2^{2^{\bigO(\width{\aprob}^8 \cdot \log\size{\aprob})}}$ for every instance $\aprob$ of this
problem.
\end{theorem}
\optproof{\subsection{Proof of Theorem \ref{thm:pce-complexity}}}{ The
  \twoexptime-hardness lower bound is given in \cite[Theorem
    18]{DBLP:conf/lpar/EchenimIP20}. The upper bound is explained in
  the proof of \cite[Theorem 32]{EIP21a}. Note that, although the
  systems are assumed to be established in \cite{EIP21a}, it is easy
  to check that only left-establishment is actually used in the
  proof. Left-establishment is used to transform the entailment into
  an restricted one (in the sense of \cite{EIP21a}). \qed}

A natural question arises in this context: which of the restrictions
from the above theorem can be relaxed and what is the price, in terms
of computational complexity, of relaxing (some of) them? In the light
of Theorem \ref{thm:right-disc-undec} below, the connectivity
restriction cannot be completely dropped. Further, if we drop the
establishment condition, the problem becomes undecidable \cite[Theorem
  6]{echenim:hal-02951630}, even if both the left/right progress and
connectivity conditions apply.

\begin{theorem}\label{thm:right-disc-undec}
  The progressing, left-connected and established entailment problem
  is undecidable. 
\end{theorem}
\optproof{\subsection{Proof of Theorem \ref{thm:right-disc-undec}}}{By
  a reduction from the known undecidable problem of universality of
  context-free languages. A context-free grammar $G =
  \tuple{N,T,S,\Delta}$ consists of a finite set $N$ of nonterminals,
  a finite set $T$ of terminals, a start symbol $S \in N$ and a finite
  set $\Delta$ of productions of the form $A \rightarrow w$, where $A
  \in N$ and $w \in (N \cup T)^*$. Given finite strings $u, v \in (N
  \cup T)^*$, the step relation $u \Rightarrow v$ replaces a
  nonterminal $A$ of $u$ by the right-hand side $w$ of a production $A
  \rightarrow w$ and $\Rightarrow^*$ denotes the reflexive and
  transitive closure of $\Rightarrow$. The language of $G$ is the set
  $\lang{G}$ of finite strings $w \in T^*$, such that $s \Rightarrow^*
  w$. The problem $T^* \subseteq \lang{G}$ is known as the
  universality problem, known to be undecidable
  \cite{BarHillel61}. W.l.o.g. we assume further
  that: \begin{compactitem}
  \item $T = \set{0,1}$, because every terminal can be encoded as a
    binary string,
  \item $\lang{G}$ does not contain the empty string $\epsilon$,
    because computing a grammar $G'$ such that $\lang{G'} = \lang{G}
    \cap T^+$ is possible \comment[np]{maybe remove this (complexity
      is irrelevant here):}\comment[ri]{but we need that the grammar
      does not produce the empty word}\comment[np]{yes} and, moreover, we can reduce
    from the modified universality problem problem $T^+ \subseteq
    \lang{G'}$ instead of the original $T^* \subseteq \lang{G}$,
  \item $G$ is in Greibach normal form, i.e.\ it contains only
    production rules of the form $A_0 \rightarrow a A_1 \ldots A_n$,
    where $A_0, \ldots A_n \in N$, for some $n \geq 0$ and $a \in T$.
  \end{compactitem}
  We use the special variables $\hat{0}$ and $\hat{1}$ to denote the
  binary digits $0$ and $1$. For each nonterminal $A_0 \in N$, we have
  a predicate $A_0(x,y,\hat{0},\hat{1})$ and a rule
  $A_0(x,y,\hat{0},\hat{1}) \Leftarrow x \mapsto (\hat{a},x_1) *
  A_1(x_1,x_2,\hat{0},\hat{1}) * \ldots A_n(x_n,y,\hat{0},\hat{1})$,
  for each rule $A_0 \rightarrow a A_1 \ldots A_n$ of $G$. Moreover,
  we consider the rules $T(x,y,\hat{0},\hat{1}) \Leftarrow x \mapsto
  (\hat{a},z) * T(z,y,\hat{0},\hat{1})$ and $T(x,y,\hat{0},\hat{1})
  \Leftarrow x \mapsto (\hat{a},y)$, for all $a \in \set{0,1}$ and let
  $\asid$ be the resulting SID. It is easy to check that the SID is
  progressing and established and that, moreover, the rules for $T$
  are connected. Finally, the entailment $\pb{\hat{0} \not\iseq
    \hat{1} * T(x,y)}{S(x,y)}{\asid}$ is valid if and only if $T^+
  \subseteq \lang{G}$. \qed}

The second decidable class of entailment problems \cite{EIP21a}
relaxes the connectivity condition and replaces the establishment with
a syntactic condition (that can be checked in polynomial time in the
size of the SID), while remaining \twoexptime-complete. Informally,
the definition forbids (dis)equations between existential variables in
symbolic heaps or rules: the only allowed (dis)equations are of the
form $x \bowtie y$ where $x$ is a free variable (viewed as a constant
in \cite{EIP21a}).  The definition given below is essentially
equivalent to that of \cite{EIP21a}, but avoids any reference to
constants; instead it uses a notion of \rpositional functions, which
helps to identify existential variables that are always replaced by a
free variable from the initial formula during unfolding.

An {\em \rpositional function}
maps every $n$-ary predicate symbol $p$ occurring in $\asid$ to a
subset of $\interv{1}{n}$. Given an \rpositional function $\fvargs$
and a formula $\phi$, we denote by $\vargs{\phi}$ the set of variables
$x_i$ such that $\phi$ contains a predicate atom $p(x_1,\dots,x_n)$
with $i \in \fvargs(p)$.  Note that $\vargf$ is stable under
substitutions, i.e.\ $\vargs{\phi\sigma} = (\vargs{\phi})\sigma$, for
each formula $\phi$ and each substitution $\sigma$.

\begin{definition}\label{def:fv-profile}
Let $\psi$ be a formula and $\asid$ be an SID. The {\em \fvprofile} of
the pair $(\psi,\asid)$ is the \rpositional function $\fvargs$ such that the sets
$\fvargs(p)$, for $p \in \allpreds$, are the maximal sets satisfying the
following conditions:
\begin{enumerate}
\item{
$\vargs{\psi} \subseteq \fv{\psi}$.
}
\item{ For all predicate symbols $p\in \preds{\psi}$, all rules
  $p(x_1,\dots,x_n) \Leftarrow \genbod$ in $\asid$, all predicate
  atoms $q(y_1,\dots,y_m)$ in $\genbod$ and all $i \in \fvargs(q)$,
  there exists $j \in \fvargs(p)$ such that $x_j = y_i$.}
\end{enumerate}
The \fvprofile of  $(\psi,\asid)$ is denoted by $\fvargsof{\psi}{\asid}$.
\end{definition}
Intuitively, given a predicate $p \in \allpreds$,
the set $\fvargsof{\psi}{\asid}(p)$ denotes the formal parameters of
$p$ that, in every unfolding of $\psi$, will always be substituted by
variables occurring freely in $\psi$. It is easy to check that
$\fvargsof{\psi}{\asid}$ can be computed in polynomial time
w.r.t.\ $\size{\psi} + \size{\asid}$, using a straightforward greatest
fixpoint algorithm. The algorithm starts with a function mapping every
predicate $p$ of arity $n$ to $\interv{1}{n}$ and repeatedly removes
elements from the sets $\fvargs(p)$ to ensure that the above
conditions hold. In the worst case, we may have eventually $\fvargs(p)
= \emptyset$ for all predicate symbols $p$.

\comment[ri]{abbreviations: P for progressing, \C for \connected, \R for \restricted}
\begin{definition}\label{def:restricted}
Let $\fvargs$ be an \rpositional function, and $V$ be a set of
variables. A formula $\phi$ is {\em \restricted} (\R) w.r.t.\ $V$ iff the
following hold: \begin{compactenum}
\item\label{it1:restricted}{for every disequation $y \not \iseq z$ in
  $\phi$, we have $\{ y,z\} \cap V \not = \emptyset$, and}
\item\label{it2:restricted}{$\vargs{\phi} \subseteq V$.}
\end{compactenum}
A rule $p(x_1,\ldots,x_n) \Leftarrow x \mapsto (y_1,\dots,y_\rank) *
\genrest$ is: \begin{compactitem}
%% \item \emph{\progressing} iff $x \in \vargs{p(x_1,\dots,x_n)} \cup
%%   \{x_1\}$ and $\genrest$ contains no occurrences of points-to atoms,
%%   \comment[np]{I do not really see what is the point of generalizing
%%     the ``progressing'' condition (since one can always switch
%%     arguments if needed). Further, I think this condition will be
%%     problematic. Thus I would simply remove this item.}
%
\item \emph{\connected} (\C) iff for every atom $q(z_1,\dots,z_m)$
  occurring in $\genrest$, we have $z_1 \in \vargs{p(x_1,\dots,x_n)}
  \cup \{y_1, \ldots, y_\rank\}$, 
\item \emph{\restricted} (\R) iff $\genrest$ is \restricted
  w.r.t. $\vargs{p(x_1,\dots,x_n)}$.
\end{compactitem}
%% A progressing rule $p(x_1,\ldots,x_n) \Leftarrow x_1 \mapsto
%% (y_1,\dots,y_\rank) * \genrest$ is {\em \connected} if for every
%% atom $q(z_1,\dots,z_m)$ occurring in $\genrest$, we have either
%% $z_1 \in \{ y_1,\dots,y_\rank \}$ or %\comment[me]{slight modif}
%% $z_1\in \vargs{p(x_1,\dots,x_n)}$.  %$z_1 = x_i$, for some $i \in
%% \fvargs(p)$.  This rule is {\em \restricted} if $\genrest$ is
%% \restricted w.r.t. $\vargs{p(x_1,\dots,x_n)}$.
An SID $\asid$ is \emph{P}  (resp. \emph{\C},
\emph{\R}) \emph{for a formula} $\phi$ iff every rule in
$\bigcup_{p \in \preds{\phi}}\asid(p)$ is P
(resp. \C, \R). 
%% An entailment problem $\pb{\phi}{\psi}{\asid}$ is \emph{left}
%% (\emph{right}) \emph{progressing} (resp. \emph{\connected},
%% \emph{\restricted}) iff $\asid$ is progressing (resp. \connected,
%% \restricted) for $\phi$ ($\psi$), where $\lambda$ is considered to
%% be $\fvargsof{\phi}{\asid}$ ($\fvargsof{\psi}{\asid}$). An
%% entailment problem is \emph{progressing} (resp. \emph{\connected},
%% \emph{\restricted}) iff it is both left and right progressing
%% (resp. \connected, \restricted).
An SID $\asid$ is \emph{\C} (\emph{\R}) \emph{for a
  formula} $\phi$ iff every rule in $\bigcup_{p \in
  \preds{\phi}}\asid(p)$ is \C (\R). An entailment
problem $\pb{\phi}{\psi}{\asid}$ is \emph{left}- (\emph{right}-)
\emph{\C}, (\emph{\R}) iff $\asid$ is \C
(\R) for $\phi$ ($\psi$), where $\lambda$ is considered to be
$\fvargsof{\phi}{\asid}$ ($\fvargsof{\psi}{\asid}$). An entailment
problem is \emph{\C} (\emph{\R}) iff it is both left-
and right-\C (\R). 
\end{definition}
\noindent
The class of progressing, \connected and \restricted entailment
problems has been shown to be a generalization of the class of
progressing, connected and left-established problems, because the
latter can be reduced to the former by a many-one reduction
\cite[Theorem 13]{EIP21a} that runs in time $\size{\aprob} \cdot
2^{\bigO(\width{\aprob}^2)}$ on input $\aprob$ (Figure
\ref{fig:reductions}) and preserves the problem's width
asymptotically.

\begin{figure}[thb]
  \vspace*{-1.5\baselineskip}
\caption{Many-one Reductions between Decidable Entailment Problems}
\label{fig:reductions}
\centerline{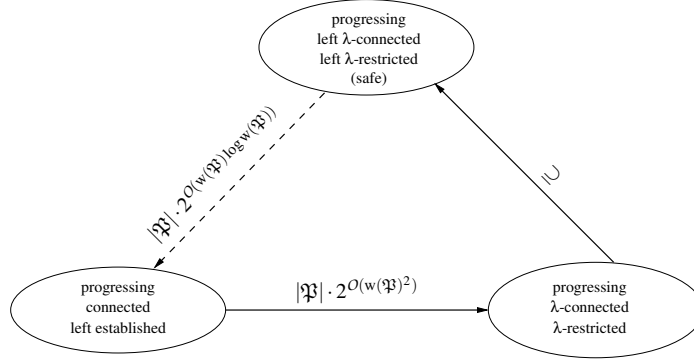}
  \vspace*{-.5\baselineskip}
\end{figure}

In the rest of this paper we close the loop by defining a syntactic
extension of \progressing, \connected and \restricted entailment
problems and by showing that this extension can be reduced to the
class of progressing, connected and left-established entailment
problems by a many-one reduction. The new fragment is defined as
follows:
\begin{definition}\label{def:safe}
An entailment problem $\pb{\phi}{\psi}{\asid}$ is {\em \safe} if, for
$\fvargs \isdef \fvargsof{\psi}{\asid}$, the following
hold: \begin{compactenum}
\item\label{safe:progress} every rule in $\asid$ is progressing,
\item\label{safe:init} $\psi$ is \restricted w.r.t.\ $\fv{\phi}$,
\item\label{safe:ind} all the rules from $\bigcup_{p \in \preds{\psi}}
  \asid(p)$ are \connected and \restricted.
\end{compactenum}
\end{definition}
Note that there is no condition on the formula $\phi$, or on the rules
defining the predicates occurring only in $\phi$, other than the
progress condition. The conditions in Definition
\ref{def:safe} ensure that all the disequations occurring in any
unfolding of $\psi$ involve at least one variable that is free in
$\phi$. Further, the heaps of the model of $\psi$ must be
\emph{forests}, i.e.\ unions of trees, the roots of which are
associated with the first argument of the predicate atoms in $\psi$ or
to free variables from $\phi$.

We refer the reader to Figure \ref{fig:reductions} for a general
picture of the entailment problems considered so far and of the
many-one reductions between them, where the reduction corresponding to
the dashed arrow is the concern of the next section. Importantly,
since all reductions are many-one, taking time polynomial in the size
and exponential in the width of the input problem, while preserving
its width asymptotically, the three classes from Figure
\ref{fig:reductions} can be unified into a single (2EXPTIME-complete)
class of entailments.

%\begin{proposition}
%\label{prop:restricted}
%Let $\pb{\phi}{\psi}{\asid}$ be a \safe entailment problem and
%$\fvargs \isdef \fvargsof{\psi}{\asid}$.  Every unfolding $\psi'$ of
%$\psi$ is \restricted w.r.t. $\fv{\phi}$.
%\end{proposition}
%\begin{proof}
%  Assume $\psi = p(u_1, \ldots, u_n) * \varphi$ and that the rule
%  $p(x_1, \ldots, x_n) \Leftarrow x_1\mapsto (y_1, \ldots, y_\rank) *
%  \genrest$ is applied with substitution $\theta \isdef
%  \substinterv{x_i}{u_i}{i}{\interv{1}{n}}$ yielding $[x_1\mapsto (y_1,
%    \ldots, y_\rank) * \genrest]\theta * \varphi$.
%  By hypothesis $\psi$ is \restricted w.r.t.\ $\fv{\phi}$ and by
%  Proposition \ref{prop:restr-V}, we deduce that $[x_1\mapsto (y_1,
%    \ldots, y_\rank) * \genrest]\theta * \varphi$ is also \restricted
%  w.r.t.\ $\vargs{\phi}$. The rest of the proof follows by an
%  induction on the length of the unfolding.
%\end{proof}

%\newcommand{\project}[2]{#1[#2]}

%If $(\astore,\aheap) \models \alpha_1 *\alpha_2$  then
%there exist $\aheap_1,\aheap_2$ such that 
%$(\astore,\aheap_i) \models \alpha_i$. 
%In the following, we will be denote by $\project{\aheap}{\alpha_i}$ arbitrary heaps such that
%$(\astore,\project{\aheap}{\alpha_i}) \models \alpha_i$ and
%$\aheap = \project{\aheap}{\alpha_1} \dunion \project{\aheap}{\alpha_2}$
%(note that such  heaps  are not unique).
%\comment[np]{slight abuse of notation here because $\project{\aheap}{\alpha_1}$ depends on $\alpha_1 * \alpha_2$}

\section{Reducing Safe to Established Entailments}
\label{sect:reduction}

\comment[ri]{added some intuition} 

In a model of a safe SID (Definition \ref{def:safe}), the existential
variables introduced by the replacement of predicate atoms with
corresponding rule bodies are not required to be allocated. This is
because safe SIDs are more liberal than established SIDs and allow
heap structures with an unbounded number of dangling pointers. As
observed in \cite{EIP21a}, checking the validity of an entailment
(w.r.t a restricted SID) can be done by considering only those
structures in which the dangling pointers point to pairwise distinct
locations. The main idea of the hereby reduction of safe to
established entailment problems is that any such structure can be
extended by allocating all dangling pointers separately and, moreover,
the extended structures can be defined by an established SID.

In what follows, we fix an arbitrary instance $\aprob =
\pb{\phi}{\psi}{\asid}$ of the \safe entailment problem (Definition
\ref{def:safe}) and denote by $\fvargs \isdef \fvargsof{\psi}{\asid}$
the fv-profile of $(\psi,\asid)$ (Definition
\ref{def:fv-profile}). Let $\mvec{\fvar} \isdef
(\fvar_1,\dots,\fvar_\nfv)$ be the vector of free variables from
$\phi$ and $\psi$, where the order of variables is not important and
assume w.l.o.g.\ that $\nfv > 0$. \comment[ri]{removed: Note that, by
  definition, the right-hand side of every rule in $\asid$ is a
  symbolic heap (i.e.\ the rules in $\asid$ contain no disjunction).}
Let $\leftpred \isdef \preds{\phi}$ and $\rightpred \isdef
\preds{\psi}$ be the sets of predicate symbols that depend on the
predicate symbols occurring in the left- and right-hand side of the
entailment, respectively. We assume that $\phi$ and $\psi$ contain no
points-to atoms and that $\leftpred \cap \rightpred =
\emptyset$. Again, these assumptions lose no generality, because a
points-to atom $u \mapsto (v_1,\dots,v_\rank)$ can be replaced by a
predicate atom $p(u,v_1,\dots,v_\rank)$, where $p$ is a fresh
predicate symbol associated with the rule $p(x,y_1,\dots,y_\rank)
\Leftarrow x \mapsto (y_1,\dots,y_\rank)$. Moreover the condition
$\leftpred \cap \rightpred \neq \emptyset$ may be enforced by
considering two copies of each predicate, for the left-hand side and
for the right-hand side, respectively. Finally, we assume that every
rule contains exactly $\maxv$ existential variables, for some fixed
$\maxv\in {\Bbb N}$; this condition can be enforced by adding dummy
literals $x \iseq x$ if needed.

We describe a reduction of $\aprob$ to an equivalent progressing,
connected, and left-established entailment problem. The reduction will
\comment[ri]{changed} extend heaps, by adding  $\nfv+\maxv$ record
fields. We shall therefore often consider heaps and points-to atoms
having $\rank+\nfv+\maxv$ record fields, where the formal definitions
are similar to those given previously. Usually such formul{\ae} and heaps
will be written with a prime. These additional record fields will be
used to ensure that the constructed system is connected, by adding all
the existential variables of a given rule (as well as the variables in
$\fvar_1,\dots,\fvar_\nfv$) into the image of the location allocated
by the considered rule. Furthermore, the left-establishment condition
will be enforced by adding predicates and rules in order to allocate
all the locations that correspond to existential quantifiers and that
are not already allocated, making such locations point to a dummy
vector $\mvec{\bot} \isdef (\bot, \ldots, \bot)$, of length
$\rank+\nfv + \maxv$, where $\bot$ is the special constant denoting
empty heap entries. To this aim, we shall use a predicate symbol
$\botp$ associated with the rule $\botp(x) \Leftarrow x \mapsto
\mvec{\bot}$, where $\mvec{\bot} = (\bot,\dots,\bot)$. \comment[me]{maybe this equality is redundant with definition of $\mvec{\bot}$ above. It looks like there is no difference between $\bot$ and $\mvec{\bot}$, maybe change the font?}\comment[np]{done} Note that
allocating all these locations will entail (by definition of the
separating conjunction) that they are distinct, thus the addition of
such predicates and rules will reduce the number of satisfiable
unfoldings. However, due to the restrictions on the use of
\comment[ri]{added footnote} disequations\footnote{Point
  (\ref{it1:restricted}) of Definition \ref{def:restricted} in
  conjunction with point (\ref{safe:init}) of Definition
  \ref{def:safe}.}, we shall see that this does not change the status
of the entailment problem.

\begin{definition}\label{def:struct-img}
For any total function $\gamma: \Loc \rightarrow \Loc$ and any tuple
$\mvec{\ell} = \tuple{\ell_1,\dots,\ell_n} \in \Loc^n$, we denote by
$\gamma(\mvec{\ell})$ the tuple
$\tuple{\gamma(\ell_1),\dots,\gamma(\ell_n)}$. If $\astore$ is a
store, then $\apl{\gamma}{\astore}$ denotes the store with domain
$\dom{\astore}$, such that $\apl{\gamma}{\astore}(x) \isdef
\gamma(\astore(x))$, for all $x \in \dom{\astore}$. Consider a heap
$\aheap$ such that for all $\ell \neq \ell' \in \dom{\aheap}$, we have
$\gamma(\ell) \neq \gamma(\ell')$. Then $\apl{\gamma}{\aheap}$ denotes
the heap with domain $\dom{\apl{\gamma}{\aheap}} = \{\gamma(\ell) \mid
\ell \in \dom{\aheap}\}$, such that
$\apl{\gamma}{\aheap}(\gamma(\ell)) \isdef \gamma(\aheap(\ell))$, for
all $\ell \in \dom{\aheap}$.
\end{definition}

\putInAppendix{\section{Additional Material for Section \ref{sect:reduction}}}{}

The following lemma identifies conditions ensuring that the
application of a mapping to a structure (Definition
\ref{def:struct-img}) preserves the truth value of a formula.
\begin{lemma}
\label{lem:gamma}
Given a set of variables $V$, let $\alpha$ be a formula that is
\restricted w.r.t. $V$, such that $\preds{\alpha}\subseteq \rightpred$
and let $(\astore,\aheap)$ be an $\asid$-model of $\alpha$. For every
mapping $\gamma: \Loc \rightarrow \Loc$ such that $\gamma(\ell) =
\gamma(\ell') \Rightarrow \ell =\ell'$ holds whenever either $\{
\ell,\ell' \} \subseteq \dom{\aheap}$ or $\{\ell,\ell' \} \cap
\astore(V) \not = \emptyset$, we have
$(\apl{\gamma}{\astore},\apl{\gamma}{\aheap}) \models_{\asid} \alpha$.
\end{lemma}
\optproof{\subsection{Proof of Lemma \ref{lem:gamma}}}
{
We first establish the following:
\begin{proposition}\label{prop:restr-V}
  Let $\fvargs$ be an \rpositional function and let $V$ be a set of
  variables. Consider an atom $p(u_1, \ldots, u_n)$ such that
  $\vargs{p(u_1, \ldots, u_n)} \subseteq V$ and a \restricted rule
  $p(x_1,\ldots,x_n) \Leftarrow \genbod$. Then the formula
  $\genbod\substmult{x_1}{u_1}{x_n}{u_n}$ is \restricted w.r.t.\ $V$.
\end{proposition}
\begin{proof}
Let $\theta = \substmult{x_1}{u_1}{x_n}{u_n}$.  By hypothesis
$\genbod$ is \restricted w.r.t. $\vargs{p(x_1,\dots,x_n)}$, thus we
have $\vargs{\genbod} \subseteq \vargs{p(x_1,\ldots, x_n)}$ and we
deduce
that \[\vargs{\genbod\theta}\ =\ \vargs{\genbod}\theta\ \subseteq\ \vargs{p(x_1,
  \ldots, x_n)}\theta\ =\ \vargs{p(u_1, \ldots, u_n)}\ \subseteq\ V.\]
Consider an atom $y\theta\not\iseq z\theta$ occurring in
$\genbod\theta$. Then necessarily $y\not\iseq z$ occurs in $\genbod$,
hence $\set{y,z} \cap \vargs{p(x_1,\dots,x_n)}\neq \emptyset$ and
$\set{y\theta, z\theta} \cap \vargs{p(u_1, \ldots, u_n)} \neq
\emptyset$. Since $\vargs{p(u_1, \ldots, u_n)} \subseteq V$, we obtain
$\set{y\theta, z\theta} \cap V \neq \emptyset$, as required. \qed
\end{proof}
The proof of Lemma \ref{lem:gamma} is by structural induction on the
definition of the relation $\models_{\asid}$. We distinguish the
following cases:
\begin{compactitem}
\item{If $\alpha = (x \iseq y)$, then $\astore(x) = \astore(y)$ and
  $\aheap = \emptyset$, thus $\apl{\gamma}{\astore}(x) =
  \apl{\gamma}{\astore}(y)$ and $\apl{\gamma}{\aheap} = \emptyset$.
  Therefore, $(\apl{\gamma}{\astore},\apl{\gamma}{\aheap})
  \models_{\asid} x \iseq y$.}
\item{If $\alpha = (x \not \iseq y)$, then $\astore(x)\not =
  \astore(y)$ and $\aheap = \emptyset$, hence $\apl{\gamma}{\aheap} =
  \emptyset$.  Since $\alpha$ is \restricted w.r.t.\ $V$, necessarily
  one of the variables $x$ or $y$ occurs in $V$, hence $\{
  \astore(x),\astore(y) \} \cap \astore(V) \not = \emptyset$. By the
  hypotheses of the lemma this entails that $\apl{\gamma}{\astore}(x)
  \not = \apl{\gamma}{\astore}(y)$.  Thus
  $(\apl{\gamma}{\astore},\apl{\gamma}{\aheap}) \models_{\asid} x \not
  \iseq y$.}
\item{If $\alpha = x \mapsto (y_1,\dots,y_\rank)$ then we have
  $\dom{\aheap} = \{\astore(x)\}$ and $\aheap(\astore(x)) =
  \tuple{\astore(y_1),\dots,\astore(y_n)}$. Therefore
  $\dom{\apl{\gamma}{\aheap}} = \{\gamma(\astore(x))\}$ and
  $\apl{\gamma}{\aheap}(\astore(x)) =
  \apl{\gamma}{\tuple{\astore(y_1),\dots,\astore(y_n)}}$, thus we
  obtain $(\apl{\gamma}{\astore},\apl{\gamma}{\aheap}) \models_{\asid}
  \alpha$.}
\item{If $\alpha = \alpha_1 * \alpha_2$ then there exists
  $\aheap_1,\aheap_2$, such that $\aheap = \aheap_1 \dunion \aheap_2$
  and $(\astore,\aheap_i) \models_{\asid} \alpha_i$, for $i=1,2$.
  Since $\alpha_i$ is \restricted w.r.t.\ $V$ and that
  $\preds{\alpha_i} \subseteq \rightpred$, by the induction
  hypothesis, we obtain that
  $(\apl{\gamma}{\astore},\apl{\gamma}{\aheap_i}) \models_{\asid}
  \alpha_i$, for $i = 1,2$. For all $\ell_i \in \dom{\aheap_i}$ with
  $i = 1,2$, we have $\ell_1,\ell_2 \in \dom{\aheap}$ and $\ell_1 \not
  = \ell_2$, thus by the hypothesis of the lemma $\gamma(\ell_1) \not
  = \gamma(\ell_2)$.  Hence $\dom{\apl{\gamma}{\aheap_1}}$ and
  $\dom{\apl{\gamma}{\aheap_2}}$ are disjoint, we have
  $\apl{\gamma}{\aheap} = \apl{\gamma}{\aheap_1} \dunion
  \apl{\gamma}{\aheap_2}$ and therefore
  $(\apl{\gamma}{\astore},\apl{\gamma}{\aheap}) \models_{\asid}
  \alpha$.}
\item{If $\alpha = p(u_1,\dots,u_n)$ then $\asid$ contains a rule
  $p(x_1,\dots,x_n) \Leftarrow \genbod$ and $(\astorep,\aheap)
  \models_{\asid} \genbod\theta$, for some extension $\astorep$ of
  $\astore$, with $\theta \isdef
  \substinterv{x_i}{u_i}{i}{\interv{1}{n}}$.  Since $\preds{\alpha}
  \subseteq \rightpred$, the rule must be \restricted by Condition
  \ref{safe:ind} in Definition \ref{def:safe}.  By Proposition
  \ref{prop:restr-V}, we deduce that $\genbod\theta$ is \restricted
  w.r.t.\ $V$.  Moreover, we have $\preds{\genbod\theta} =
  \preds{\genbod} \subseteq \rightpred$, thus by the induction
  hypothesis $(\apl{\gamma}{\astorep},\apl{\gamma}{\aheap})
  \models_{\asid} \genbod\theta$, and therefore
  $(\apl{\gamma}{\astore},\apl{\gamma}{\aheap}) \models_{\asid}
  \alpha$, because $\apl{\gamma}{\astorep}$ is an extension of
  $\apl{\gamma}{\astore}$.} \qed
\end{compactitem}}

If $\gamma$ is, moreover, injective, then the result of Lemma
\ref{lem:gamma} holds for any formula:
\begin{lemma}
\label{lem:gamma_inj}
Let $\alpha$ be a formula and let $(\astore,\aheap)$ be an
$\asid$-model of $\alpha$.  For every injective mapping $\gamma: \Loc
\rightarrow \Loc$ we have
$(\apl{\gamma}{\astore},\apl{\gamma}{\aheap}) \models_{\asid} \alpha$.
\end{lemma}
\begin{proof}
  Since the implication $\gamma(\ell) = \gamma(\ell') \Rightarrow \ell
  =\ell'$ holds for all $\ell,\ell'\in\Loc$. \qed
\end{proof}

\subsection{Expansions and Truncations}
\label{subsect:exp}

We introduce a so-called \emph{expansion} relation on structures, as
well as a \emph{truncation} operation on heaps. Intuitively, the
expansion of a structure is a structure with the same store and whose
heap is augmented with new allocated locations (each pointing to
$\vec{\bot}$) and additional record fields, referring in particular to
all the newly added allocated locations. These locations are
introduced to accommodate all the existential variables of the
predicate-less unfolding of the left-hand side of the entailment (to
ensure that the obtained entailment is left-established).  Conversely,
the truncation of a heap is the heap obtained by removing these extra
locations. We also introduce the notion of a $\gamma$-expansion which
is a structure whose image by $\gamma$ is an expansion.

\comment[ri]{added this paragraph only once, it should be enough to
  remember who $\maxv$ and $\nfv$ are below} We recall that,
throughout this and the next sections (\S\ref{sect:cons} and
\S\ref{sect:ant}), $\mvec{\fvar} = (\fvar_1, \ldots, \fvar_\nfv)$
denotes the vector of free variables occurring in the problem, which
is assumed to be fixed throughout this section and that $\{
\fvar_1,\dots,\fvar_\nfv, \bot \}\subseteq \dom{\astore}$, for every
store $\astore$ considered here. Moreover, we assume w.l.o.g.\ that
$\fvar_1,\dots,\fvar_\nfv$ do not occur in the considered SID $\asid$
and denote by $\maxv$ the number of existential variables in each rule
of $\asid$. We refer to Figure \ref{fig:expansion} for an illustration
of the definition below: 

\begin{figure}[t!]
  \caption{Heap Expansion and Truncation}
  \label{fig:expansion}
  \centerline{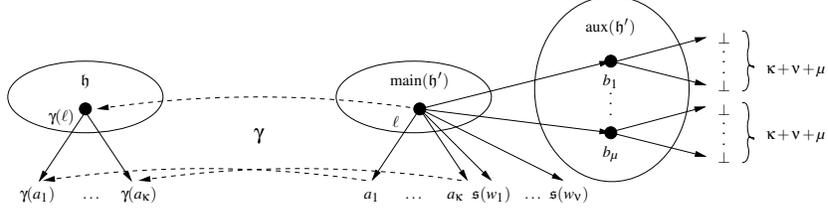}
  \vspace*{-\baselineskip}
\end{figure}

\begin{definition}\label{def:decor_heap}
  Let $\gamma: \Loc \rightarrow \Loc$ be a total mapping.  A structure
  $(\astore,\aheap')$ is a {\em $\gamma$-expansion} (or simply an {\em
    expansion} if $\gamma = \id$) of some structure
  $(\astore,\aheap)$, denoted by $(\astore,\aheap') \expands{\gamma}
  (\astore,\aheap)$, if \comment[ri]{added, makes things more clear
    for later} $\aheap : \Loc \rightarrow \Loc^\rank$, $\aheap' : \Loc
  \rightarrow \Loc^{\rank+\maxv+\nfv}$ and there exist two disjoint
  heaps, $\mpart{\aheap'}$ and $\bpart{\aheap'}$, such that $\aheap' =
  \mpart{\aheap'} \dunion \bpart{\aheap'}$ and the following
  hold: \begin{compactenum}
  \item{for all $\ell_1,\ell_2 \in \dom{\mpart{\aheap'}}$, if
    $\gamma(\ell_1) = \gamma(\ell_2)$ then $\ell_1 =
    \ell_2$,\label{it:decor:inj}}
  \item{$\gamma(\dom{\mpart{\aheap'}}) =
    \dom{\aheap}$,\label{it:decor:dom}}
  \item{for each $\ell\in \dom{\mpart{\aheap'}}$, we have
    $\aheap'(\ell) =
    \tuple{\mvec{a},\astore(\mvec{\fvar}),b_1^{\ell},\dots,b_{\maxv}^{\ell}}$,
    for some locations $b_1^{\ell}, \ldots, b_\maxv^{\ell} \in \Loc$ and
    $\gamma(\mvec{a}) = \aheap(\gamma(\ell))$,\label{it:decor:img} }
  \item{for each $\ell \in \dom{\bpart{\aheap'}}$, we have
    $\aheap'(\ell) = \mvec{\bots}$ and there exists a location
    $\ell'\in \dom{\mpart{\aheap'}}$ such that
    $\mpart{\aheap'}(\ell')$ is of the form
    $\tuple{\mvec{a},\mvec{\ell},b_1^{\ell'},\dots,b_{\maxv}^{\ell'}}$
    where $\vec{\ell}$ is a tuple of locations and $\ell =
    b_i^{\ell'}$, for some $i \in \interv{1}{\maxv}$.  The element
    $\ell'$ is called the \emph{\gconnection of $\ell$ in $\aheap'$}
    and is denoted by $\gcon{\ell}{\aheap'}$.\footnote{Note that
      $\ell'$ does not depend on $\gamma$, and if several such
      locations exist, then one is chosen arbitrarily.}
    \label{it:decor:link}}
  \end{compactenum}
\end{definition}
Let $(\astore,\aheap')$ be a $\gamma$-expansion of $(\astore,\aheap)$
and let $\ell \in \dom{\mpart{\aheap'}}$ be a location. Since $\nfv
>0$ and for all $i \in \interv{1}{\nfv}$, $\astore(\fvar_i)$ occurs in
$\aheap'(\ell)$, and since we assume that $\astore(\fvar_i) \not
=\bots = \astore(\bot)$ for every $i \in \interv{1}{\nfv}$,
necessarily $\mpart{\aheap'}(\ell) \neq \mvec{\bots}$.  This entails
that the decomposition $\aheap' = \mpart{\aheap'} \dunion
\bpart{\aheap'}$ is unique: $\mpart{\aheap'}$ and $\bpart{\aheap'}$
are the restrictions of $\aheap'$ to the locations $\ell$ in
$\dom{\aheap'}$ such that $\aheap'(\ell) \not = \mvec{\bots}$ and
$\aheap'(\ell) = \mvec{\bots}$, respectively. In the following, we
shall thus freely use the notations $\bpart{\aheap'}$ and
$\mpart{\aheap'}$, for arbitrary heaps $\aheap'$.

\begin{definition}\label{def:trunc}
  Given a heap $\aheap'$, we denote by $\trunc{\aheap'}$ the heap
  $\aheap$ defined as follows: $\dom{\aheap} \isdef \dom{\aheap'}
  \setminus \{\ell\in \dom{\aheap'} \mid \aheap'(\ell) =
    \mvec{\bots}\}$ and for all $\ell \in \dom{\aheap}$, if
  $\aheap'(\ell) = (\ell_1, \ldots, \ell_{\rank+\nfv+\maxv})$, then
  $\aheap(\ell) \isdef (\ell_1, \ldots, \ell_\rank)$.
\end{definition}
\comment[ri]{added} Note that, if $\aheap = \trunc{\aheap'}$ then
$\aheap : \Loc \rightarrow \Loc^\rank$ and $\aheap' : \Loc \rightarrow
\Loc^{\rank+\maxv+\nfv}$ are heaps of different out-degrees. In the
following, we silently assume this fact, to avoid cluttering the
notation by explicitly specifying the out-degree of a heap.

\begin{example}
  Assume that $\Loc = {\Bbb N}$, $\nfv = \maxv = 1$.
  Let $\astore$ be a store such that $\astore(\fvar_1) = 0$. We consider:
  \begin{eqnarray*}
    \aheap & \isdef & \{ \tuple{1,2},\tuple{2,2} \},\\
    \aheap'_1 & \isdef & \{ \tuple{1,(2,0,1)}, \tuple{2,(2,0,3)}, \tuple{3,(\bot,\bot,\bot)} \},\\
    \aheap'_2 & \isdef & \{ \tuple{1,(3,0,1)}, \tuple{2,(4,0,3)}, \tuple{3,(\bot,\bot,\bot)} \}.
  \end{eqnarray*}
 We have $(\astore, \aheap'_1) \expands{\id} (\astore, \aheap)$ and
 $(\astore, \aheap'_2) \expands{\gamma} (\astore, \aheap)$, with
 $\gamma \isdef \{ \tuple{1,1}, \tuple{2,2}, \tuple{3,2}, \tuple{4,2}
 \}$.  Also, $\trunc{\aheap_1'} = \{ \tuple{1,2}, \tuple{2,2} \} =
 \aheap$ and $\trunc{\aheap_2'} = \{ \tuple{1,3}, \tuple{2,4} \}$.
 Note that $\aheap$ has out-degree $\rank = 1$, whereas $\aheap_1'$
 and $\aheap_2'$ have out-degree $3$.  \hfill$\blacksquare$
\end{example}
\putInAppendix{\section{Additional Material for Section \ref{subsect:exp}}}
{
%% \begin{proposition}\label{prop:expand-id}
%% For all stores $\astore$ and functions $\gamma$, we have $(\astore,
%% \emptyset) \expands{\gamma} (\astore, \emptyset)$.
%% \end{proposition}
%% \begin{proof}
%% Immediate.
%% \end{proof}
  \begin{lemma}\label{prop:truncemp}
    Consider a total mapping $\gamma: \Loc \rightarrow \Loc$, a store
    $\astore$ and two heaps $\aheap$ and $\aheap'$, such that $(\astore,
    \aheap')\expands{\gamma} (\astore, \aheap)$. We have $\aheap' =
    \emptyset$ if and only if $\aheap = \emptyset$.
  \end{lemma}
  \begin{proof}
    Since $(\astore, \aheap')\expands{\gamma} (\astore, \aheap)$,
    following the notations of Definition \ref{def:decor_heap},
    $\aheap'$ is of the form $\mpart{\aheap'} \uplus
    \bpart{\aheap'}$. First assume $\aheap' = \emptyset$. Then
    necessarily ${\mpart{\aheap'}} = \emptyset$ and by Condition
    \ref{it:decor:dom} of Definition \ref{def:decor_heap}, we deduce
    that ${\aheap} = \emptyset$. Conversely, if $\aheap = \emptyset$,
    then by Condition \ref{it:decor:dom} of Definition
    \ref{def:decor_heap} we must have $\mpart{\aheap'} = \emptyset$
    and by Condition \ref{it:decor:link}, we also have
    $\bpart{\aheap'} = \emptyset$. \qed
  \end{proof}
  
  \begin{lemma}\label{prop:expand-dunion}
    Consider a store $\astore$, heaps $\aheap_1, \aheap_1', \aheap_2,
    \aheap_2'$ and total mappings $\gamma, \gamma_1, \gamma_2: \Loc
    \rightarrow \Loc$, such that the following
    hold:\begin{compactenum}
    \item\label{it1:expand-union} $\dom{\aheap_1'}\cap \dom{\aheap_{2}'}
      = \emptyset$ and $\dom{\aheap_1}\cap \dom{\aheap_2} = \emptyset$.
    \item\label{it2:expand-union} $(\astore, \aheap_i')
      \expands{\gamma_i} (\astore, \aheap_i)$, for $i = 1,2$,
    \item\label{it3:expand-union} for all $\ell \in \locs{\aheap_i'}$,
      $\gamma(\ell) = \gamma_i(\ell)$, for $i = 1,2$.    
    \end{compactenum}
    By letting $\aheap'\isdef \aheap_1'\dunion \aheap_2'$ and
    $\aheap\isdef \aheap_1\dunion \aheap_2$, we have
    $\aheap'\expands{\gamma}\aheap$.
  \end{lemma}
  \begin{proof}
    By Point (\ref{it2:expand-union}), we have $\aheap_i' =
    \mpart{\aheap_i'}\dunion \bpart{\aheap_{i}'}$, for $i = 1,2$. By
    Point (\ref{it1:expand-union}), $\dom{\mpart{\aheap_1'}} \cap
    \dom{\mpart{\aheap_2'}} = \dom{\bpart{\aheap_1'}} \cap
    \dom{\bpart{\aheap_2'}} = \emptyset$.  Let $\mpart{\aheap'} =
    \mpart{\aheap_1'} \dunion \mpart{\aheap_2'}$ and $\bpart{\aheap'}
    = \bpart{\aheap_1'} \dunion \bpart{\aheap_2'}$. We prove that
    these heaps satisfy the conditions of Definition
    \ref{def:decor_heap}: \begin{compactenum}
    \item Let $\ell_1, \ell_2 \in \dom{\mpart{\aheap'}}$. If
      $\ell_1\in \dom{\mpart{\aheap_1'}}$ and $\ell_2\in
      \dom{\mpart{\aheap_2'}}$, then $\gamma(\ell_1) =
      \gamma_1(\ell_1) \in \dom{\aheap_1}$, and $\gamma(\ell_2) =
      \gamma_2(\ell_2) \in \dom{\aheap_2}$. Since both sets are
      disjoint, it is impossible to have $\gamma(\ell_1) =
      \gamma(\ell_2)$. Else, if $\ell_1, \ell_2 \in
      \dom{\mpart{\aheap_1'}}$, then $\gamma(\ell_1) =
      \gamma_1(\ell_1)$ and $\gamma(\ell_2) = \gamma_1(\ell_2)$, thus
      $\gamma_1(\ell_1) = \gamma_1(\ell_2)$ and $\ell_1 = \ell_2$
      follows from point (\ref{it2:expand-union}). The other cases are
      symmetric.
    \item We compute:
      \begin{eqnarray*}
        \gamma(\dom{\mpart{\aheap'}}) &=& \gamma(\dom{\mpart{\aheap_1'}} \dunion \dom{\mpart{\aheap_{2}'}})\\
        & =& \gamma_1(\dom{\mpart{\aheap_1'}}) \dunion \gamma_2(\dom{\mpart{\aheap_2'}})\\
        & = & \dom{\aheap_1} \dunion \dom{\aheap_2}\\
        & =& \dom{\aheap}.
      \end{eqnarray*}
    \item Let $\ell \in \dom{\mpart{\aheap'}}$ and assume $\ell \in
      \dom{\mpart{\aheap_1'}}$; the other case is symmetric. Then by
      construction $\mpart{\aheap'}(\ell) = \mpart{\aheap_1'}(\ell)$
      is of the form
      $\tuple{\mvec{a},\astore(\mvec{\fvar}),b_1^{\ell},\dots,b_{\maxv}^{\ell}}$
      where $\gamma_1(\mvec{a}) = \aheap_1(\gamma_1(\ell))$. Since
      $\mvec{a}, \ell \in \locs{\aheap_1'}$, we deduce that
      $\gamma(\mvec{a}) = \aheap_1(\gamma(\ell))$, hence
      $\gamma(\mvec{a}) = \aheap(\gamma(\ell))$.
    \item Let $\ell \in \dom{\bpart{\aheap'}}$. By definition, $\ell
      \in \dom{\bpart{\aheap_i'}}$, for some $i = 1,2$, and by Point
      \ref{it2:expand-union}, we deduce that $\aheap'_i(\ell) =
      \mvec{\bots}$ and that there exists a \gconnection $\ell'\in
      \dom{\mpart{\aheap_i'}}$ of $\ell$ in $\aheap'_i$.  Then
      $\aheap'(\ell) = \aheap'_i(\ell) = \mvec{\bots}$ and $\ell'\in
      \dom{\mpart{\aheap'}}$, hence $\ell'$ is also a \gconnection of
      $\ell$ in $\aheap'$. \qed
    \end{compactenum}
  \end{proof}
  
  \begin{lemma}\label{prop:expand-bij}
    Assume $(\astore, \aheap') \expands{\gamma} (\astore, \aheap)$ and
    let $\eta: \Loc\rightarrow \Loc$ be a bijection such that
    $\eta(\bots) = \bots$. If $\gamma' = \gamma \circ\inv{\eta}$, then
    $(\apl{\eta}{\astore}, \apl{\eta}{\aheap'}) \expands{\gamma'}
    (\apl{\eta}{\astore},{\aheap})$.
  \end{lemma}
  \begin{proof}
    Note that $\apl{\eta}{\aheap'}$ is well-defined, since $\eta$ is a
    bijection.  We show that $\apl{\eta}{\mpart{\aheap'}}$ and
    $\apl{\eta}{\bpart{\aheap'}}$ satisfy the conditions of Definition
    \ref{def:decor_heap}, thus proving the result: \begin{compactenum}
    \item Let $\ell, \ell' \in \dom{\apl{\eta}{\mpart{\aheap'}}}$ and
      assume that $\gamma'(\ell) = \gamma'(\ell')$. Then there exist
      $\ell_1, \ell_2 \in \dom{\mpart{\aheap'}}$ such that $\ell =
      \eta(\ell_1)$ and $\ell' = \eta(\ell_2)$, hence
      \[\gamma(\ell_1)\ =\ \gamma\circ\inv{\eta}(\eta(\ell_1))\ =\ \gamma'(\ell)\ =\ 
      \gamma'(\ell')\ =\ \gamma\circ\inv{\eta}(\eta(\ell_2))\ =\ \gamma(\ell_2),\]
      so that $\ell_1 = \ell_2$, %\comment[ri]{added:} 
      by Point
      (\ref{it:decor:inj}) of Definition \ref{def:decor_heap}, leading
      to $\ell= \ell'$.
    \item We compute:
      \[
      \begin{array}{lllll}
        \gamma'(\dom{\apl{\eta}{\mpart{\aheap'}}}) & = & \gamma'(\eta(\dom{\mpart{\aheap'}})) 
        & = & \gamma(\dom{\mpart{\aheap'}}) \\
        & = & \dom{\aheap}.
      \end{array}\]
    \item Let $\ell \in \dom{\apl{\eta}{\mpart{\aheap'}}}$. Then there
      exists $\ell_1\in \dom{\mpart{\aheap'}}$ such that $\ell =
      \eta(\ell_1)$, and $\mpart{\aheap'}(\ell_1)$ is of the form
      $\tuple{\mvec{a},\astore(\mvec{\fvar}),b_1,\dots,b_{\maxv}}$,
      for some $b_i\in \Loc$ ($1 \leq i \leq \maxv$) and
      $\gamma(\mvec{a}) = \aheap(\gamma(\ell_1))$. Hence,
      $\apl{\eta}{\mpart{\aheap'}}(\ell) =
      \tuple{\eta(\mvec{a}),\apl{\eta}{\astore}(\mvec{\fvar}),\eta(b_1),\dots,\eta(b_{\maxv})}$
      and
      \[\gamma'(\eta(\mvec{a}))\ =\ \gamma(\mvec{a})\ =\ \aheap(\gamma(\ell_1))\ =\ \aheap(\gamma'(\ell)).\]
    \item Let $\ell \in \dom{\apl{\eta}{\bpart{\aheap'}}}$. Then there
      exists $\ell_2\in \dom{\bpart{\aheap'}}$ such that $\ell =
      \eta(\ell_2)$, and since $\eta(\bot) = \bot$, we have
      $\apl{\eta}{\bpart{\aheap'}}(\ell) = \eta(\bpart{\aheap'}(\ell_2))
      = \eta(\mvec{\bots}) = \mvec{\bots}$. By hypothesis $\ell_2$
      admits a \gconnection $\ell'$ in $\aheap'$, and it is
      straightforward to check that $\eta(\ell')$ is a \gconnection of
      $\ell$ in $\eta(\aheap')$. \qed
    \end{compactenum}
  \end{proof}
  
  \begin{lemma}\label{prop:expand-restr}
    Assume $(\astore, \aheap') \expands{\id} (\astore, \aheap)$, let
    $D\subseteq \dom{\aheap'}$ and consider $\aheap_1'$
    (resp. $\aheap_1$), the restriction of $\aheap'$ (resp. $\aheap$) to
    $D$. If every location in $\dom{\bpart{\aheap'}} \cap D$ has a
    \gconnection in $\aheap_1'$, then $(\astore, \aheap_1')
    \expands{\id} (\astore, \aheap_1)$.
  \end{lemma}
  \begin{proof}
    We check the conditions of Definition
    \ref{def:decor_heap}: \begin{compactenum}
    \item Trivial. 
    \item Since $\dom{\mpart{\aheap'}} = \dom{\aheap}$, we obtain
      $\dom{\mpart{\aheap'_1}} = \dom{\mpart{\aheap'}} \cap D = \dom{\aheap}
      \cap D = \dom{\aheap_1}$. 
    \item Since Point (\ref{it:decor:img}) holds for $\mpart{\aheap'}$, 
      it also holds for $\mpart{\aheap'_1}$. 
    \item Let $\ell \in \dom{\bpart{\aheap'_1}} = \dom{\bpart{\aheap'}}
      \cap D$.  We have $\aheap'_1(\ell) = \aheap'(\ell) =
      \mvec{\bots}$.  Further, $\ell$ has a \gconnection in $\aheap_1'$
      by hypothesis. \qed
    \end{compactenum}
  \end{proof}

  \begin{lemma}\label{prop:trunctunion}
    Let $\aheap_1',\aheap_2'$ be disjoint heaps and let $\aheap_i \isdef
    \trunc{\aheap_i'}$, for $i = 1,2$.  Then $\aheap_1$ and $\aheap_2$
    are disjoint and $\trunc{\aheap_1' \dunion \aheap_2'} = \aheap_1
    \dunion \aheap_2$.
  \end{lemma}
  \begin{proof}
    By definition, $\dom{\aheap_i} \subseteq \dom{\aheap_i'}$ hence
    since $\aheap_1',\aheap_2'$ are disjoint, $\aheap_1$ and $\aheap_2$
    are also disjoint. 
    and $\aheap \isdef \aheap_1 \dunion \aheap_2$.  Let $D$ be the set
    of locations $\ell$ such that $\aheap_1'(\ell) = \mvec{\bots}$ or
    $\aheap_2'(\ell) = \mvec{\bots}$.  By Definition \ref{def:trunc}, we
    have $\dom{\aheap_i} = \dom{\aheap_i'} \setminus D$, hence
    $\dom{\aheap_1 \dunion \aheap_2} = \dom{\aheap_1' \dunion \aheap_2'}
    \setminus D = \dom{\trunc{\aheap_1' \dunion \aheap_2'}}$.  It is
    clear that this entails that $\trunc{\aheap_1' \dunion \aheap_2'} =
    \aheap_1 \dunion \aheap_2$. \qed
  \end{proof}
}

\begin{lemma}\label{prop:exp-trunc}
  If $(\astore, \aheap') \expands{\gamma} (\astore, \aheap)$ then
  $\aheap = \apl{\gamma}{\trunc{\aheap'}}$, hence
  $(\astore, \aheap') \expands{\id} (\astore, \trunc{\aheap'})$.
\end{lemma}
\optproof{\subsection{Proof of Lemma \ref{prop:exp-trunc}}}
{
  Since $(\astore, \aheap') \expands{\gamma} (\astore, \aheap)$, the
  restriction of $\gamma$ to $\dom{\mpart{\aheap'}}$ is injective and
  $\gamma(\dom{\mpart{\aheap'}}) = \dom{\aheap}$, by Point
  (\ref{it:decor:dom}) of Definition \ref{def:decor_heap}.
  Furthermore, by Definition \ref{def:trunc}, $\dom{\trunc{\aheap'}} =
  \dom{\mpart{\aheap'}}$.  Thus $\dom{\apl{\gamma}{\trunc{\aheap'}}} =
  \gamma(\dom{\trunc{\aheap'}}) = \gamma(\dom{\mpart{\aheap'}}) =
  \dom{\aheap}$.  Moreover, for any $\ell \in \dom{\mpart{\aheap'}}$,
  we have $\aheap'(\ell) = \tuple{\vec{a}, \astore(\vec{w}), b_1^\ell,
    \ldots, b_\maxv^\ell}$ and $\aheap(\gamma(\ell)) =
  \gamma(\vec{a})$, for some $\vec{a} \in \Loc^\rank$ and $b_1^\ell,
  \ldots, b_\maxv^\ell \in \Loc$, thus by Definition \ref{def:trunc},
  $\gamma(\trunc{\aheap'}(\ell)) = \gamma(\vec{a}) =
  \gamma(\aheap(\ell))$.
  The second part of the proof is immediate. \qed
}

The converse of Lemma \ref{prop:exp-trunc} does not hold in general,
but it holds under some additional conditions:
\begin{lemma}
\label{prop:trunc}
Consider a store $\astore$, let $\aheap'$ be a heap and let $\aheap
\isdef \trunc{\aheap'}$. Let $D_2 \isdef \{ \ell\in\dom{\aheap'} \mid
\aheap'(\ell) = \mvec{\bots} \}$ and $D_1 \isdef \dom{\aheap'}
\setminus D_2$. Assume that:
\begin{compactenum}
\item{for every location $\ell \in D_1$, $\aheap(\ell)$ is of the form
  $(\ell_1,\dots,\ell_\rank)$ and $\aheap'(\ell)$ is of the form
  $(\ell_1,\dots,\ell_\rank,\astore(\mvec{\fvar}),\ell_1',\dots,\ell_\maxv')$; \label{trunc:img}}
\item{every location $\ell \in D_2$ has a \gconnection in
  $\aheap'$. \label{trunc:link}}
\end{compactenum}
Then $(\astore,\aheap') \expands{\id} (\astore,\aheap)$.
\end{lemma}
\optproof{\subsection{Proof of Lemma \ref{prop:trunc}}}
{
It is straightforward to check that Conditions \ref{it:decor:inj} and
\ref{it:decor:dom} of Definition \ref{def:decor_heap} hold, with
$\mpart{\aheap'}$ (resp. $\bpart{\aheap'}$) defined as the restriction
of $\aheap'$ to $D_1$ (resp. $D_2$).  Condition \ref{it:decor:img}
follows immediately from Point \ref{trunc:img} and from the definition
of $\trunc{\aheap'}$.  Condition \ref{it:decor:link} holds by Point
\ref{trunc:link}. \qed
}

\subsection{Transforming the Consequent}
\label{sect:cons}

\newcommand{\rightdecor}[1]{\widehat{#1}}
\newcommand{\rd}[1]{\widehat{#1}}

\newcommand{\decoration}{decoration\xspace}
\newcommand{\decorated}[2]{{#1}_{#2}}

\newcommand{\leftsid}{\asid_l}
\newcommand{\rightsid}{\asid_r}

\newcommand{\astorei}{\hat{\astore}}

We first describe the transformation for the right-hand side of the
entailment problem, as this transformation is simpler. 

\begin{definition} \label{def:right_deco}
  We associate each $n$-ary predicate $p\in \rightpred$ with a new
  predicate $\rd{p}$ of arity $n + \nfv$.  We denote by
  $\rightdecor{\alpha}$ the formula obtained from $\alpha$ by
  replacing every predicate atom $p(x_1,\dots,x_n)$ by
  $\rd{p}(x_1,\dots,x_n,\mvec{\fvar})$, \comment[ri]{added} where
  $\mvec{\fvar} = (\fvar_1, \ldots, \fvar_\nfv)$.
\end{definition} 
%% \comment[me]{Maybe we should explain here why we have $\mvec{w}$ as a
%%   parameter in the rules instead of an arbitrary vector
%%   $\mvec{\omega}$, that it would clutter the presentation since we
%%   want $\mvec{\omega}$ to be systematically mapped to
%%   $\mvec{w}$. Also, I added indices to $\xi$ and $\chi$, stating that
%%   they may be omitted if necessary.}\comment[np]{I have added a remark below, please check/modif/expand}

\begin{definition}\label{def:right_rule}
We denote by $\rightdecor{\asid}$ the set of rules of the form: 
\[\rd{p}(x_1,\dots,x_n,\mvec{\fvar}) \Leftarrow  
x_1 \mapsto (y_1,\dots,y_\rank,\mvec{\fvar},z_1,\dots,z_{\maxv})\sigma
* \rightdecor{\genrest}\sigma * \xi_I * \chi_\sigma\]
where: \begin{compactitem}
\item{ $p(x_1,\dots,x_n) \Leftarrow x_1 \mapsto (y_1,\dots,y_\rank) *
  \genrest$ is a rule in $\asid$ with $p \in \rightpred$, }
\item{ $\{ z_1,\dots,z_\maxv \}$ is a set of variables disjoint from $\fv{\genrest} \cup \{
  x_1,\dots,x_n,\allowbreak y_1,\dots,y_\rank,\allowbreak \fvar_1,\dots,\fvar_\nfv \}$, }
\item{ $\sigma$ is a substitution with $\dom{\sigma} \subseteq
  \fv{\genrest} \setminus \{ x_1 \}$ and $\img{\sigma} \subseteq \{
  \fvar_1,\dots,\fvar_\nfv \}$, }
\item{ $\xi_I \isdef \bigast_{i \in I} \botp(z_i)$, with $I \subseteq
  \{ 1,\dots,\maxv\}$, }
\item{ $\chi_\sigma \isdef \bigast_{x\in \dom{\sigma}} x \iseq
  x\sigma$.  }
\end{compactitem}
We denote by $\rightsid$ the set of rules in $\rightdecor{\asid}$ that
are connected\footnote{Note that all the rules in $\rightdecor{\asid}$
  are progressing.}.
\end{definition} 
Note that the free variables $\mvec{w}$ are added as parameters in the
rules above, instead of some arbitrary tuple of fresh variables
$\mvec{\omega}$, \comment[ri]{added} of the same length as
$\mvec{w}$. This is for the sake of conciseness \comment[ri]{instead
  of clarity}, since these parameters $\mvec{\omega}$ will be
systematically mapped to $\mvec{w}$.

\begin{example}
Assume that $\psi = \exists x~.~p(x,\fvar_1)$, with $\nfv = 1$, $\maxv = 1$ and $\fvargs(p) = \{ 2 \}$.
Assume also that $p$ is associated with the rule:
\[
\begin{tabular}{lll}
$p(u_1,u_2)$ & $\Leftarrow$ & $u_1 \mapsto u_1 * q(u_2)$. \\
\end{tabular}
\]
Observe that the rule is \connected, but not connected. 
Then $\dom{\sigma} \subseteq \set{u_2}$, $\img{\sigma} \subseteq \set{\fvar_1}$ and $I\subseteq \set{1}$, so that $\rightdecor{\asid}$ contains the following rules:
\[
\begin{tabular}{lllll}
(1) & $p(u_1,u_2,\fvar_1)$ & $\Leftarrow$ & $u_1 \mapsto (u_1,\fvar_1,z_1) * q(u_2)$ \\
(2) & $p(u_1,u_2,\fvar_1)$ & $\Leftarrow$ & $u_1 \mapsto (u_1,\fvar_1,z_1) * q(u_2) * \botp(z_1) $ \\
(3) & $p(u_1,u_2,\fvar_1)$ & $\Leftarrow$ & $u_1 \mapsto (u_1,\fvar_1,z_1) * q(\fvar_1) * u_2 \iseq \fvar_1$ \\
(4) & $p(u_1,u_2,\fvar_1)$ & $\Leftarrow$ & $u_1 \mapsto (u_1,\fvar_1,z_1) * q(\fvar_1) * \botp(z_1) * u_2 \iseq \fvar_1$ \\
\end{tabular}
\]
Rules (1) and (2) are not connected, hence do not occur in
$\rightsid$.  Rules (3) and (4) are connected, hence occur in
$\rightsid$. Note that (4) is established, but (3) is not.
\hfill$\blacksquare$
\end{example}

We now relate the SIDs $\asid$ and $\rightsid$ by the following result: 

\begin{lemma}\label{lem:right_equiv}
Let $\alpha$ be a formula that is \restricted w.r.t.\ $\{
\fvar_1,\dots,\fvar_\nfv\}$ and contains no points-to atoms, with
$\preds{\alpha} \subseteq \rightpred$. Given a store $\astore$ and two
heaps $\aheap$ and $\aheap'$, such that $(\astore,\aheap')
\expands{\id} (\astore,\aheap)$, we have $(\astore,\aheap')
\models_{\rightsid} \rightdecor{\alpha}$ if and only if
$(\astore,\aheap) \models_{\asid} \alpha$.
\end{lemma}

\putInAppendix{\section{Additional Material for Section \ref{sect:cons}}}{}
\optproof{\subsection{Proof of Lemma \ref{lem:right_equiv}}}{
  We first need the following:
  
  \begin{lemma}\label{prop:rightgconn}
    If $(\astore,\aheap') \models_{\rightsid} \rightdecor{\alpha}$
    then for all $\ell \in \dom{\bpart{\aheap'}}$, $\ell$ has a
    \gconnection in $\aheap'$.
  \end{lemma}
  \begin{proof}
    Let $\ell \in \dom{\bpart{\aheap'}}$. By definition $\aheap'(\ell) =
    \mvec{\bots}$.  Observe that $\ell$ cannot be allocated by
    the points-to atom of a rule in $\rightsid$,
    since otherwise, by Definition \ref{def:right_rule}, it would be
    mapped to a tuple containing
    $\astore(\fvar_1),\dots,\astore(\fvar_\nfv)$, hence to a tuple
    distinct from $\mvec{\bot}$ since $\nfv > 0$ and $\astore(\fvar_i)
    \not = \bots$ for $i = 1,\dots,\nfv$.  Consequently, $\ell$ must be
    allocated by a predicate $\botp(z_i)$ invoked in a rule in Definition
    \ref{def:right_rule}.  Since $z_i$ also occurs as one of the last
    $\maxv$ components on the right-hand side of the points-to atom of the
    considered rule, necessarily $\ell$ has a \gconnection in $\aheap'$. \qed
  \end{proof}

The proof of Lemma \ref{lem:right_equiv} is by induction on the pair
$(\len{\aheap},\size{\alpha})$, using the lexicographic order.  We
distinguish several cases, depending on the form of
$\alpha$: \begin{compactitem}
\item{If $\alpha$ is of the form $x \iseq y$ then by Definition
  \ref{def:right_deco}, $\rightdecor{\alpha} = \alpha$. We have
  $(\astore,\aheap') \models_{\rightsid} \rightdecor{\alpha}$ iff
  $\aheap' = \emptyset$ and $\astore(x) = \astore(y)$, and
  $(\astore,\aheap) \models_{\asid} {\alpha}$ iff $\aheap = \emptyset$
  and $\astore(x) = \astore(y)$. By Lemma \ref{prop:truncemp},
  $\aheap' = \emptyset$ iff $\aheap = \emptyset$, hence the result.}
\item{The proof is similar if $\alpha$ is of the form $x \not\iseq
  y$.}
\item{Assume that $\alpha = \alpha_1 \vee \alpha_2$.  By construction
  we have $\rightdecor{\alpha} = \rightdecor{\alpha}_1 \vee
  \rightdecor{\alpha}_2$.  Now, $(\astore,\aheap') \models_{\rightsid}
  \rightdecor{\alpha}$ if and only if $(\astore,\aheap')
  \models_{\rightsid} \rightdecor{\alpha_i}$, for some $i \in
  \{1,2\}$.  By the induction hypothesis, this is equivalent to
  $(\astore,\aheap) \models_{\rightsid} \alpha_i$, for some $i \in
  \{1,2\}$, i.e.\ equivalent to $(\astore,\aheap) \models_{\rightsid}
  \alpha$. }
\item{Assume that $\alpha = \alpha_1 * \alpha_2$. Then it is
  straightforward to check that $\rightdecor{\alpha} =
  \rightdecor{\alpha}_1 * \rightdecor{\alpha}_2$.  If
  $(\astore,\aheap') \models_{\rightsid} \rightdecor{\alpha}$ then
  there exists $\aheap_1',\aheap_2'$ such that $\aheap' = \aheap_1'
  \dunion \aheap_2'$ and $(\astore,\aheap'_i) \models_{\rightsid}
  \rightdecor{\alpha_i}$ for $i = 1,2$.  Let $\aheap_i =
  \trunc{\aheap_i'}$.  By Lemma \ref{prop:rightgconn}, every location
  in $\bpart{\aheap_i'}$ has a \gconnection in $\aheap_i'$, thus, by
  Lemma \ref{prop:expand-restr} (applied with $D = \dom{\aheap_i'}$),
  we deduce that $(\astore,\aheap_i') \expands{\id}
  (\astore,\aheap_i)$. By the induction hypothesis, we deduce that
  $(\astore,\aheap_i) \models_{\asid} \alpha_i$, and by Lemma
  \ref{prop:trunctunion}, $\aheap = \aheap_1 \dunion \aheap_2$.  Thus
  $(\astore,\aheap) \models_{\asid} \alpha$.

  Conversely, assume that $(\astore,\aheap) \models_{\asid} \alpha$.
  Then there exists $\aheap_1,\aheap_2$ such that $\aheap = \aheap_1
  \dunion \aheap_2$ and $(\astore,\aheap_i) \models_{\asid} \alpha_i$
  for $i = 1,2$. We define the sets: 
  \[\begin{array}{rcl}
  D & \isdef & \dom{\aheap'} \setminus \dom{\aheap} \\
  D_1 & \isdef & \setof{d \in D}{\gcon{d}{\aheap'} \in \dom{\aheap_1}} \\ 
  D_2 & \isdef & D \setminus D_1
  \end{array}\] 
  For $i = 1,2$, let $\aheap_i'$ be the restriction of $\aheap'$ to
  $\dom{\aheap_i}\cup D_i$. Since $(\astore, \aheap') \expands{\id}
  (\astore, \aheap)$ by hypothesis, it is straightforward to verify
  that $(\astore, \aheap_i') \expands{\id} (\astore, \aheap_i)$, and
  by the induction hypothesis, we deduce that $(\astore,\aheap_i')
  \models_{\rightsid} \rightdecor{\alpha_i}$ for $i = 1,2$.  By
  construction, $D_1$, $D_2$, $\dom{\aheap_1}$ and $\dom{\aheap_2}$
  are pairwise disjoint and $\dom{\aheap_1} \cup D_1 \cup
  \dom{\aheap_2} \cup D_2 = \dom{\aheap} \cup D = \dom{\aheap'}$,
  hence $\aheap' = \aheap_1' \dunion \aheap_2'$. We conclude that
  $(\astore,\aheap') \models_{\rightsid} \rightdecor{\alpha}$.  }
\item Assume that $\alpha = p(u_1,\dots,u_n)$, so that
  $\rightdecor{\alpha} = \rd{p}(u_1,\dots,u_n,\mvec{\fvar})$.  If
  $(\astore,\aheap') \models_{\rightsid} \rightdecor{\alpha}$, then
  $\rightsid$ contains a rule of the form
  \[ \rd{p}(x_1,\dots,x_n,\mvec{\fvar}) \Leftarrow x_1 \mapsto (y_1,\dots,y_\rank,\mvec{\fvar},z_1,\dots,z_\maxv) *
  \rightdecor{\genrest}\sigma * \xi_I * \chi_\sigma,\] satisfying the
  conditions of Definition \ref{def:right_rule}, and there exists an
  extension $\astorep$ of $\astore$ such that $(\astorep,\aheap')
  \models_{\rightsid} u_1 \mapsto
  (y_1,\dots,y_\rank,\mvec{\fvar},z_1,\dots,z_\maxv)\theta *
  \rightdecor{\genrest}\sigma\theta * \xi_I\theta * \chi_\sigma\theta$,
  with $\theta \isdef \substinterv{x_i}{u_i}{i}{\interv{1}{n}}$.  Then,
  necessarily, $\asid$ contains a rule:
  \[p(x_1,\dots,x_n) \Leftarrow x_1 \mapsto (y_1,\dots,y_\rank) * \genrest\ (\ddagger)\]
  Let $\aheap_1'$ be the restriction of $\aheap'$ to $\dom{\aheap}
  \setminus \left(\{ \astore(u_1) \} \cup \{\astorep(z_i) \mid i \in I
  \}\right)$ and let $\aheap_1 \isdef \trunc{\aheap_1'}$. By
  definition of $\aheap_1'$, we have $(\astorep,\aheap_1')
  \models_{\rightsid} \rightdecor{\genrest}\sigma\theta$.  By
  hypothesis $\alpha$ is \restricted w.r.t.\ $\{
  \fvar_1,\dots,\fvar_\nfv\}$, thus $\vargs{p(u_1,\dots,u_n)}
  \subseteq \{ \fvar_1,\dots,\fvar_\nfv\}$ and, since the entailment
  problem under consideration is \safe, rule $(\ddagger)$ is
  necessarily \restricted by Point (\ref{safe:ind}) of Definition
  \ref{def:safe}. Thus, $\genrest\theta$ must be \restricted
  w.r.t.\ $\{ \fvar_1,\dots,\fvar_\nfv\}$, by Lemma
  \ref{prop:restr-V}.  Since the image of $\sigma$ is contained in $\{
  \fvar_1,\dots,\fvar_\nfv\}$, we deduce that $\genrest\sigma\theta$
  is also \restricted w.r.t.\ $\{ \fvar_1,\dots,\fvar_\nfv\}$.  By the
  induction hypothesis, this entails that $(\astorep,\aheap_1)
  \models_{\asid} \genrest\sigma\theta$.  Since $(\astorep,\emptyset)
  \models_{\asid} \chi_\sigma\theta$ by definition of $\chi_\sigma$
  --- see Definition \ref{def:right_rule}, we have $\astorep(x\theta)
  = \astorep(x\sigma\theta)$ for every variable $x \in
  \dom{\sigma}$. But the latter equality trivially holds for every
  variable $x \not\in \dom{\sigma}$, hence replacing all variables
  $x\sigma\theta$ occurring in $\genrest\sigma\theta$ by $x\theta$
  preserves the truth value of the formula in
  $(\astorep,\aheap_1)$. Consequently $\genrest\sigma\theta$ and
  $\genrest\theta$ have the same truth value in $(\astorep,\aheap_1)$,
  and thus $(\astorep,\aheap_1) \models_{\asid} \genrest\theta$.
 
  Let $\aheap_u'$ denote the restriction of $\aheap'$ to
  $\set{\astore(u_1)}$.  By construction $\aheap'(\astorep(z_i)) =
  \mvec{\bots}$ for every $i \in I$, hence by Lemmas
  \ref{prop:exp-trunc} and \ref{prop:trunctunion} we have
  \[\aheap = \trunc{\aheap'} = \trunc{\aheap_1' \dunion \aheap_u'} =
  \aheap_1 \dunion \trunc{\aheap_u'}.\] Since $\trunc{\aheap_u'} =
  \{(\astorep(u_1), \tuple{\astorep(y_1\theta), \dots,
    \astorep(y_\rank\theta)})\}$, we deduce that
  $(\astorep,\aheap) \models_{\asid} u_1 \mapsto
  (y_1,\dots,y_\rank)\theta * \genrest\theta$, and therefore that
  $(\astore,\aheap) \models_{\asid} p(u_1,\dots,u_n)$.

  Conversely, assuming that $(\astore,\aheap) \models_{\asid}
  p(u_1,\dots,u_n)$, $\asid$ contains a rule:
  \[p(x_1,\dots,x_n) \Leftarrow x_1 \mapsto (y_1,\dots,y_\rank) *
  \genrest\ (\dagger)\] and there exists an extension $\astorep$ of
  $\astore$ such that $(\astorep,\aheap) \models_{\asid} u_1 \mapsto
  (y_1,\dots,y_\rank)\theta * \genrest\theta$, where $\theta \isdef
  \substinterv{x_i}{u_i}{i}{\interv{1}{n}}$.  Thus we must have:
  \[\aheap'(\astorep(u_1))\ =\ \tuple{\astorep(y_1\theta), \dots, \astorep(y_\rank\theta),
    \astorep(\mvec{\fvar}),\ell_1,\dots,\ell_{\maxv}}\] for some
  locations $\ell_1,\dots,\ell_{\maxv}$.  Since the variables
  $z_1,\dots,z_\maxv$ do not occur in $\{ x_1,\dots,x_n,$ \\ 
  $y_1,\ldots,y_\rank, \fvar_1, \dots, \fvar_\nfv \} \cup \fv{\genrest}$
  by hypothesis (see Definition \ref{def:right_rule}), we may assume,
  w.l.o.g., that $\astorep(z_i) = \ell_i$.  Let $I \isdef \{ i \in
  \interv{1}{\maxv}\mid \aheap'(\ell_i) = \mvec{\bot} \}$ and let
  $\sigma$ be the substitution defined as follows: \begin{compactitem}
  \item $\dom{\sigma} = (\fv{\genrest} \setminus \{ x_1 \}) \cap \{x
    \mid \astorep(x\theta) \in \{\astorep(\fvar_1), \ldots,
    \astorep(\fvar_\maxv)\}\}$, and
  \item for every variable $x \in \fv{\genrest} \setminus \{ x_1 \}$,
    such that $\astorep(x\theta) = \astorep(\fvar_i)$ for some $i \in
    \interv{1}{\nfv}$, we let $\sigma(x) \isdef w_i$; if several such
    values of $i$ are possible, then one is chosen arbitrarily.
  \end{compactitem}
  By construction (see again Definition \ref{def:right_rule}),
  $\rightdecor{\asid}$ contains the rule:
  \[ \rd{p}(x_1,\dots,x_n,\mvec{\fvar}) \Leftarrow x_1 \mapsto (y_1,\dots,y_\rank,\mvec{\fvar},z_1,\dots,z_\maxv) *
  \rightdecor{\genrest}\sigma * \xi_I * \chi_\sigma.\] Let $\aheap_1$ be
  the restriction of $\aheap$ to $\dom{\aheap}\setminus \{ \astore(u_1)
  \}$, and let $\aheap_1'$ be the restriction of $\aheap'$ to
  $\dom{\aheap'}\setminus \set{\astore(u_1), \ell_i \mid i \in I}$.
  Since $(\astore, \aheap') \expands{\id} (\astore, \aheap)$, no element
  $\ell_i$ with $i\in I$ can be in $\dom{\aheap}$.  By Definition, every
  element of $\dom{\bpart{\aheap_1'}}$ has a \gconnection in $\aheap_1'$
  (since the locations $\ell_i$ with $i \in I$ are the only elements of
  $\dom{\bpart{\aheap'}}$ whose \gconnection is $\astore(u_1)$), and by
  Lemma \ref{prop:expand-restr} we deduce that $(\astore, \aheap_1')
  \expands{\id} (\astore, \aheap_1)$.
  
  Note that by hypothesis $\alpha$ is \restricted w.r.t.\ $\{
  \fvar_1,\dots,\fvar_\nfv\}$ and the entailment problem $\aprob$ is
  \safe, thus $\genrest\theta$ is \restricted w.r.t.\ $\{
  \fvar_1,\dots,\fvar_\nfv\}$ by Lemma \ref{prop:restr-V}.  By
  construction $\img{\sigma} \subseteq \{ \fvar_1,\dots,\fvar_\nfv\}$
  thus necessarily, $\genrest\sigma\theta$ must also be \restricted
  w.r.t.\ $\{ \fvar_1,\dots,\fvar_\nfv\}$. Furthermore, since the rule
  ($\dagger$) is \connected, for every $q(x_1',\dots,x_m')$ occurring
  in $\genrest$, if $x_1'\not \in \{ y_1,\dots,y_\rank \}$ then
  necessarily $x_1'\theta \in \vargs{p(u_1, \ldots, u_n)} \subseteq \{
  \fvar_1,\dots,\fvar_\nfv \}$.  By the definition of $\sigma$ we
  deduce that $x_1' \in \dom{\sigma}$, and the rule above must be
  connected and occur in $\rightsid$ (note that we cannot have $x_1' =
  x_1$, because all the rules are progressing, hence
  $\astore_e(x_1'\theta) \in \dom{\aheap_1}$ and by definition
  $\astore(x_1\theta) \not\in \dom{\aheap_1}$).

  Now $\genrest\sigma\theta$ is \restricted w.r.t.\ $\{
  \fvar_1,\dots,\fvar_\nfv\}$, and since $(\astorep,\aheap_1)
  \models_{\asid} {\genrest}\theta$, by definition of $\sigma$ we have
  $(\astorep,\aheap_1) \models_{\asid} {\genrest}\sigma\theta$. By the
  induction hypothesis, $(\astorep,\aheap_1') \models_{\rightsid}
  \rd{\genrest}\sigma\theta$.  For every $i \in I$, we have
  $\aheap'(\astorep(z_i)) = \mvec{\bot}$, and by definition of
  $\sigma$ we have $\astorep(x\theta) = \astorep(x\sigma\theta)$ for
  every $x\in\dom{\sigma}$, thus $(\astorep,\aheap')
  \models_{\rightsid} x_1 \mapsto
  (y_1,\dots,y_\rank,\mvec{\fvar},z_1,\dots,z_\maxv)\theta *
  \genrest\sigma\theta * \xi_I\theta * \chi_\sigma\theta$, hence
  $(\astore,\aheap') \models_{\rightsid}
  \rd{p}(u_1,\dots,u_n,\mvec{\fvar})$. \qed
\end{compactitem}}

\subsection{Transforming the Antecedent}
\label{sect:ant}

\putInAppendix{\section{Additional Material for Section \ref{sect:ant}}}{}

We now describe the transformation operating on the left-hand side of
the entailment problem. For technical convenience, we make the
following assumption:

\begin{assumption}\label{ass:rule-genbod}
  We assume that, for every predicate $p \in\leftpred$, every rule of
  the form $p(x_1,\dots,x_n) \Leftarrow \genbod$ in $\asid$ and every
  atom $q(x'_1,\dots,x'_m)$ occurring in $\genbod$, $x'_1 \not \in \{
  x_1,\ldots,x_n\}$.
\end{assumption}
This is without loss of generality, because every variable $x'_1 \in
\{x_1,\ldots,x_n\}$ can be replaced by a fresh variable $z$, while
conjoining the equational atom $z \iseq x'_1$ to $\genbod$. Note that
the obtained SID may no longer be connected, but this is not
problematic, \comment[ri]{added} because the left-hand side of the
entailment is not required to be connected anyway.

\begin{definition}\label{def:deco}
We associate each pair $(p,X)$, where $p \in \leftpred$, $\arity{p}=n$
and $X \subseteq \interv{1}{n}$, with a fresh predicate symbol
$\decorated{p}{X}$, such that $\arity{\decorated{p}{X}} = n + \nfv$.
A {\em decoration} of a formula $\alpha$ containing no points-to
atoms, such that $\preds{\alpha} \subseteq \leftpred$, is a formula
obtained by replacing each predicate atom $\beta \isdef
q(y_1,\dots,y_m)$ in $\alpha$ by an atom of the form
$\decorated{q}{X_\beta}(y_1,\dots,y_m,\mvec{\fvar})$, with $X_\beta
\subseteq \interv{1}{m}$. The set of {\decoration}s of a formula
$\alpha$ is denoted by $\decor{\alpha}$.
\end{definition}
The role of the set $X$ in a predicate atom
$p_X(x_1,\ldots,x_n,\mvec{\fvar})$ will be explained below.  Note that
the set of {\decoration}s of an atom $\alpha$ is always finite.

\begin{definition}\label{def:deco_rule}
  We denote by $\decor{\asid}$ the set of rules of the form
  \[\decorated{p}{X}(x_1,\dots,x_n,\mvec{\fvar}) \Leftarrow
  x_1 \mapsto
  (y_1,\dots,y_\rank,\mvec{\fvar},z_1,\dots,z_{\maxv})\sigma *
  \genrest' * \bigast_{i \in I} \botp(z_i),\]
  where: \begin{compactitem}
  \item{ $p(x_1,\dots,x_n) \Leftarrow x_1 \mapsto (y_1,\dots,y_\rank)
    * \genrest$ is a rule in $\asid$ and $X \subseteq \interv{1}{n}$;}
  \item{$\{ z_1,\dots,z_\maxv \} = (\fv{\genrest} \cup \{
    y_1,\dots,y_\rank\}) \setminus \{ x_1,\dots,x_n \}$,}
  \item{$\sigma$ is a substitution, with $\dom{\sigma} \subseteq
    \{ z_1,\dots,z_\maxv \}$ and $\img{\sigma} \subseteq \{
    x_1,\dots,x_n,\fvar_1,\dots,\fvar_\nfv, z_1,\dots,z_\maxv \}$;}
  \item{ $\genrest'$ is a \decoration of $\genrest\sigma$;}
  \item{$I \subseteq \{ 1,\dots,\maxv\}$ and $z_i \not
    \in\dom{\sigma}$, for all $i \in I$.}
 \end{compactitem}
\end{definition}
 
\begin{lemma}\label{lem:deco2form}
  Let $\alpha$ be a formula containing no points-to atom, with
  $\preds{\alpha} \subseteq \leftpred$, and let $\alpha'$ be a
  \decoration of $\alpha$.  If $(\astore,\aheap')
  \models_{\decor{\asid}} \alpha'$ and $(\astore,\aheap')
  \expands{\id} (\astore,\aheap)$, then $(\astore,\aheap)
  \models_{\asid} \alpha$.
\end{lemma}
\optproof{
\subsection{Proof of Lemma \ref{lem:deco2form}}}
{ We need the following lemma, similar to Lemma \ref{prop:rightgconn},
  but for the left-hand side SID $\leftsid$:
%% \begin{proposition}\label{prop:decor-conj-sep}
%%   If an atom $\alpha'$ is a decoration of $\alpha$ and is of the form
%%   $\alpha_1'*\alpha_2'$, then $\alpha$ is of the form
%%   $\alpha_1*\alpha_2$ and for $i = 1,2$, $\alpha_i'$ is a decoration
%%   of $\alpha_i$.
%% \end{proposition}
%% \begin{proof}
%% Immediate.
%% \end{proof}
%% \comment[ri]{this is used only in one place and can be internalized}\comment[np]{indeed  :)}
\begin{lemma}\label{prop:leftgconn}
If $(\astore,\aheap') \models_{\leftsid} \alpha$ then for all $\ell
\in \dom{\bpart{\aheap'}}$, $\ell$ has a \gconnection in $\aheap'$.
\end{lemma}
\begin{proof}
The proof is similar to that of Lemma \ref{prop:rightgconn}. \qed
\end{proof}

The proof of Lemma \ref{lem:deco2form} is by induction on the pair
$(\len{\aheap},\size{\alpha'})$, using the lexicographic order. We
distinguish several cases: \begin{compactitem}
  \item{If $\alpha'$ is of the form $x \iseq y$ or $x \not \iseq y$,
    then necessarily $\alpha = \alpha'$ and $\aheap' =
    \emptyset$. Thus $\aheap = \emptyset$ by Lemma
    \ref{prop:truncemp} and $(\astore,\aheap) \models_{\asid}
    \alpha$.}
  \item{If $\alpha'$ is of the form $\alpha_1' \vee \alpha_2'$ then
    $\alpha$ is of the form $\alpha_1 \vee \alpha_2$ where $\alpha_i'$
    is a \decoration of $\alpha_i$.  If $(\astore,\aheap')
    \models_{\decor{\asid}} \alpha'$ then $(\astore,\aheap')
    \models_{\decor{\asid}} \alpha'_i$ for some $i = 1,2$, and by the
    induction hypothesis we deduce that $(\astore,\aheap)
    \models_{\asid} \alpha_i$, thus $(\astore,\aheap) \models_{\asid}
    \alpha$. }
  \item{If $\alpha'$ is of the form $\alpha_1' * \alpha_2'$ then
    $\aheap' = \aheap_1' \dunion \aheap_2'$, with $(\astore,\aheap_i')
    \models_{\decor{\asid}} \alpha_i'$ and $\alpha$ is of the form
    $\alpha_1 * \alpha_2$ where $\alpha_i'$ is a \decoration of
    $\alpha_i$, for $i = 1, 2$. Let $\aheap_i$ be the restriction of
    $\aheap$ to the locations occurring in $\dom{\aheap_i'}$.  By
    Lemma \ref{prop:leftgconn}, every element of
    $\dom{\bpart{\aheap_i'}}$ has a \gconnection in $\aheap_i'$.
    Therefore, by Lemma \ref{prop:expand-restr}, we deduce that
    $(\astore,\aheap_i') \expands{\id} (\astore,\aheap_i)$ and by the
    induction hypothesis we deduce that $(\astore,\aheap_i)
    \models_{\asid} \alpha_i$.  By definition of $\expands{\id}$, we
    have $\dom{\aheap_i} = \dom{\aheap_i'} \setminus \{ \ell \mid
    \aheap'(\ell) = \mvec{\bots} \}$, and since $\dom{\aheap_1'} \cap
    \dom{\aheap_2'} = \emptyset$, $\aheap_1$ and $\aheap_2$ are
    disjoint. Furthermore, we have:
    \[
    \begin{tabular}{lll}
      $\dom{\aheap}$ & $=$ & $\dom{\aheap'}\setminus \{ \ell \mid \aheap'(\ell) = \mvec{\bots}  \}$ \\
      & = & $(\dom{\aheap_1'} \cup \dom{\aheap_2'}) \setminus \{ \ell \mid \aheap'(\ell) = \mvec{\bots}  \} $\\
      &  = & $(\dom{\aheap_1'} \setminus \{ \ell \mid \aheap'(\ell) = \mvec{\bots}  \}) \cup (\dom{\aheap_2'} \setminus \{ \ell \mid \aheap'(\ell) = \mvec{\bots}  \})$ \\
      & = & $\dom{\aheap_1} \cup \dom{\aheap_2}$,
    \end{tabular}
    \]
    therefore $\aheap = \aheap_1 \dunion \aheap_2$ and
    $(\astore,\aheap) \models_{\asid} \alpha_1 * \alpha_2 = \alpha$. }
  \item{If $\alpha'$ is of the form
    $\decorated{p}{X}(u_1,\dots,u_n,\mvec{\fvar})$, then $\alpha =
    p(u_1,\dots,u_n)$.  By definition $\decor{\asid}$ contains a rule
    \[\decorated{p}{X}(x_1,\dots,x_n,\mvec{\fvar}) \Leftarrow  
    x_1 \mapsto
    (y_1,\dots,y_\rank,\mvec{\fvar},z_1,\dots,z_{\maxv})\sigma *
    \genrest' * \bigast_{i \in I} \botp(z_i)\] satisfying the
    conditions of Definition \ref{def:deco_rule}, and there exists an
    extension $\astorep$ of $\astore$ with $(\astorep,\aheap')
    \models_{\decor{\asid}} u_1 \mapsto
    (y_1,\dots,y_\rank,\mvec{\fvar},z_1,\dots,z_{\maxv})\sigma\theta *
    \genrest'\theta * \bigast_{i \in I} \botp(z_i)$, where $\theta
    \isdef \substinterv{x_i}{u_i}{i}{\interv{1}{n}}$.  In particular,
    $\asid$ contains a rule $p(x_1,\dots,x_n) \Leftarrow x_1 \mapsto
    (y_1,\dots,y_\rank) * \genrest$, where $\genrest'$ is a
    \decoration of $\genrest\sigma$ and note that $\genrest'\theta$ is
    a \decoration of $\genrest\sigma\theta$.
			
    Let $\aheap_1'$ be the restriction of $\aheap'$ to $\dom{\aheap'}
    \setminus \left( \set{\astore(u_1)} \cup \{\astorep(z_i) \mid i
    \in I \}\right)$ and let $\aheap_1$ be the restriction of $\aheap$
    to $\dom{\aheap} \setminus \{ \astore(u_1) \}$. We have
    $(\astorep,\aheap_1') \models_{\decor{\asid}} \genrest'\theta$,
    with $\len{\aheap_1'} < \len{\aheap'}$.  By Lemma
    \ref{prop:leftgconn}, every element of $\dom{\bpart{\aheap_1'}}$
    has a \gconnection in $\aheap_1'$.  Since $(\astore,\aheap')
    \expands{\id} (\astore,\aheap)$, by Lemma \ref{prop:expand-restr}
    we have $(\astorep,\aheap_1') \expands{\id} (\astorep,\aheap_1)$.
    Hence, by the induction hypothesis, we deduce that
    $(\astorep,\aheap_1) \models_{\asid} \genrest\sigma\theta$.
    Moreover we have
    \[\aheap'(u_1)\ =\ (\astorep(y_1\sigma\theta),\dots, \astore(y_\rank\sigma\theta), 
    \astore(\mvec{\fvar}),\astore(z_1\sigma\theta),\dots,\astore(z_\maxv\sigma\theta)),\] 
    and since $(\astore,\aheap') \expands{\id} (\astore,\aheap)$, we
    deduce that $\aheap(u_1) =
    (\astorep(y_1\sigma\theta),\dots,\astore(y_\rank\sigma\theta))$.
    Consequently, $(\astorep,\aheap) \models_{\asid} (x_1 \mapsto
    (y_1,\dots,y_\rank) * \genrest)\theta$, and therefore
    $(\astore,\aheap) \models_{\asid} p(x_1\theta,\dots,x_n\theta) =
    p(u_1,\dots,u_n)$. \qed}		
  \end{compactitem}
}

At this point, the set $X$ for predicate symbol $p_X$ is of little
interest: atoms are simply decorated with arbitrary sets. However, we
shall restrict the considered rules in such a way that for every model
$(\astore,\aheap)$ of an atom
$\decorated{p}{X}(x_1,\ldots,x_{n+\nfv})$, with $n = \arity{p}$, the
set $X$ denotes a set of indices $i \in \interv{1}{n}$ such that
$\astore(x_i) \in \dom{\aheap}$.  In other words, $X$ will denote a
set of formal parameters of $\decorated{p}{X}$ that are allocated in
every model of $\decorated{p}{X}$.

\newcommand{\eqf}[1]{\sim_{#1}}
\newcommand{\allocf}[1]{\mathit{Alloc}(#1)}
\newcommand{\reachf}[1]{\mathit{Reach}(#1)}
\newcommand{\welldefined}{well-defined\xspace}
 
\begin{definition}\label{def:well-defined}
Given a formula $\alpha$, we define the set $\allocf{\alpha}$ as
follows: $x \in \allocf{\alpha}$ iff  $\alpha$ contains either a
points-to atom of the form \comment[ri]{changed} $x \mapsto (y_1,
\dots, y_{\rank+\maxv+\nfv})$, or a predicate atom
$\decorated{q}{X}(x'_1,\dots,x'_{m+\nfv})$ with $x'_i = x$ for some $i
\in X$.
\end{definition}
\comment[ri]{added} Note that, in contrast with Definition
\ref{def:alloc}, we do not consider that $x \in \allocf{\alpha}$, for
those variables $x$ related to a variable from $\allocf{\alpha}$ by
equalities.

\begin{definition}
\label{def:welldefined}
A rule $\decorated{p}{X}(x_1,\dots,x_{n+\nfv}) \Leftarrow \genbod$ in
$\decor{\asid}$ with $n = \arity{p}$ with $\rho = x_1 \mapsto
(y_1,\dots,y_k,\mvec{\fvar},z_1,\dots,z_\maxv) *\rho'$ is {\em
  \welldefined} if the following conditions hold: \begin{compactenum}
\item{ $\{ x_1 \} \subseteq
  \allocf{\decorated{p}{X}(x_1,\dots,x_{n+\nfv})} \subseteq
  \allocf{\genbod}$;
\label{welldefined:propagate}}
\item{ $\fv{\genbod} \subseteq \allocf{\genbod} \cup \{
  x_1,\dots,x_{n+\nfv}\}$. \label{welldefined:established}}
\end{compactenum}
We denote by $\leftsid$ the set of \welldefined rules in
$\decor{\asid}$.
\end{definition}
 
We first establish some important properties of $\leftsid$.

\begin{lemma}
\label{prop:alloc}
If $i \in X$ then $x_i$ is allocated in every predicate-less unfolding
of $p_X(x_1, \ldots, x_{n+\nfv})$.
\end{lemma}
\optproof{\subsection{Proof of Lemma \ref{prop:alloc}}}
{
  Let $\phi$ be a predicate-less unfolding of
  $\decorated{p}{X}(x_1,\dots,x_{n+\nfv})$. The proof is by induction
  on the length of derivation from
  $\decorated{p}{X}(x_1,\dots,x_{n+\nfv})$ to $\phi$.  Assume that
  $i\in X$.  Then $\leftsid$ contains a
  rule \[\decorated{p}{X}(x_1,\dots,x_n,\mvec{\fvar}) \Leftarrow x_1
  \mapsto (y_1,\dots,y_\rank,\mvec{\fvar},z_1,\dots,z_{\maxv})\sigma *
  \genrest' * \bigast_{i \in I} \botp(z_i)\] and $x_1 \mapsto
  (y_1,\dots,y_\rank,\mvec{\fvar},z_1,\dots,z_{\maxv})\sigma *
  \genrest' * \bigast_{i \in I} \botp(z_i) \Leftarrow_{\leftsid}^*
  \psi$. If $i = 1$ then it is clear that $x_i$ is allocated in
  $\psi$. Otherwise by Condition \ref{welldefined:propagate} of
  Definition \ref{def:welldefined} we have
  \begin{eqnarray*}	
    x_i &\in &\allocf{\decorated{p}{X}(x_1,\dots,x_n,\mvec{\fvar})}\\
    & \subseteq& \allocf{x_1 \mapsto (y_1,\dots,y_\rank,\mvec{\fvar},z_1,\dots,z_{\maxv})\sigma* \genrest' * \bigast_{i \in I} \botp(z_i)} \\
    &=& \{ x_1 \} \cup \allocf{\genrest' * \bigast_{i \in I} \botp(z_i)},	
  \end{eqnarray*} 
  thus (since $\{ z_1,\dots,z_{\maxv} \} \cap \{ x_1,\dots,x_n \} =
  \emptyset$) there exists an atom $\decorated{q}{Y}(x_1', \ldots,
  x_m')$ occurring in $\genrest'$ and an index $j\in Y$ such that $x_i
  = x_j'$.  Then $\psi$ is of the form (modulo AC) $\psi' * \psi''$,
  with $\decorated{q}{Y}(x_1', \ldots, x_m') \Leftarrow_{\asid}^*
  \psi'$, and by the induction hypothesis, $x_j'$ is allocated in
  $\psi'$, hence $x_i$ is allocated in $\psi$. \qed
}

\comment[ri]{changed the corollary into a lemma}
\begin{lemma}\label{cor:pce}
  Every rule in $\leftsid$ is progressing, connected and established.
\end{lemma}
\begin{proof}
  Since $\asid$ is progressing by hypothesis, it is straightforward to
  verify that $\leftsid$ is also progressing. Consider a rule
  $\decorated{p}{X}(x_1,\dots,x_n,\mvec{\fvar}) \Leftarrow \genbod$,
  with
  \[\genbod\ \isdef\  x_1 \mapsto (y_1,\dots,y_\rank,\mvec{\fvar},z_1,\dots,z_{\maxv})\sigma * 
  \genrest' * \bigast_{i \in I} \botp(z_i),\] that occurs in
  $\leftsid$, and a predicate atom $\alpha$ occurring in $\genrest' *
  \bigast_{i \in I} \botp(z_i)$.  This rule is obtained from a rule
  $p(x_1,\dots,x_n) \Leftarrow x_1 \mapsto (y_1,\dots,y_\rank) *
  \genrest$ in $\asid$. The atom $\alpha$ is either of the form
  $\botp(z_i)$ for some $i\in I$ (so that $z_i \not \in
  \dom{\sigma}$), or a \decoration
  $\decorated{q}{Y}(x'_1,\dots,x'_m,\mvec{\fvar})$ of some atom
  $q(x'_1,\dots,x'_m)$ occurring in $\genrest$. By Assumption
  \ref{ass:rule-genbod} $x_1'\notin \set{x_1, \ldots, x_n}$, hence by
  definition of $z_1, \ldots z_\maxv$ we have $x_1'\in \set{z_1,
    \ldots, z_\maxv}\sigma$ and the rule is connected. Let $x$ be a
  variable occurring in $\fv{\genbod}\setminus\set{x_1, \ldots,
    x_{n+\nfv}}$ and assume $(\astore, \aheap) \models_{\leftsid}
  \genbod$. Then by Condition \ref{welldefined:established} of
  Definition \ref{def:welldefined}, $x\in \allocf{\genbod}$. Since $x \not = x_1$ and 
  $x_1
    \mapsto
    (y_1,\dots,y_\rank,\mvec{\fvar},z_1,\dots,z_{\maxv})\sigma$ is the only points-to atom in $\pi$, 
  $\genbod$ contains %either a points-to atom of the form $x\mapsto (\dots)$, 
  \comment[ri]{the only points-to atom in $\genbod$ is $x_1
    \mapsto
    (y_1,\dots,y_\rank,\mvec{\fvar},z_1,\dots,z_{\maxv})\sigma$,
    right?}\comment[np]{true, changed the proof accordingly} %or 
    a predicate atom
  $\decorated{q}{X}(x'_1,\dots,x'_{m+\nfv})$ where $x'_i = x$ for some
  $i \in X$. %In the first case, it is clear that $x$ is allocated in
  %any predicate-free unfolding of $\genbod$.  In the second case, 
  By
  Lemma \ref{prop:alloc}, $x_i'$ is allocated in any predicate-free
  unfolding of $\decorated{q}{X}(x'_1,\dots,x'_{m+\nfv})$, hence $x$
  is allocated in any predicate-free unfolding of $\genbod$.  This
  proves that the rule is established. \qed
\end{proof}

\newcommand{\consistent}{consistent\xspace}
\newcommand{\sfrm}[2]{#1[#2]}

We now relate the systems $\asid$ and $\leftsid$ by the following result: 

\begin{definition}
A store $\astore$ is {\em quasi-injective} if, for all $x,y\in
\dom{\astore}$, the implication $\astore(x) = \astore(y) \Rightarrow x
= y$ holds whenever $\{ x,y \} \not \subseteq \{
\fvar_1,\dots,\fvar_\nfv \}$.
\end{definition}

\begin{lemma}\label{lem:form2deco}
Let $L$ be an infinite subset of $\Loc$.  Consider a formula $\alpha$
containing no points-to atom, with $\preds{\alpha} \subseteq
\leftpred$, and let $(\astore,\aheap)$ be an $\asid$-model of
$\alpha$, where $\astore$ is quasi-injective, and $(\img{\astore} \cup
\locs{\aheap}) \cap L = \emptyset$.  There exists a \decoration
$\alpha'$ of $\alpha$, a heap $\aheap'$ and a mapping $\gamma: \Loc
\rightarrow\Loc$ such that: \begin{compactitem}
\item $(\astore,\aheap') \expands{\gamma} (\astore,\aheap)$,
\item if $\ell \not \in L$ then $\gamma(\ell) = \ell$,
\item $\locs{\aheap'} \setminus \img{\astore} \subseteq L$,
\item $\dom{\bpart{\aheap'}} \subseteq L$ and
\item $(\astore,\aheap') \models_{\leftsid} \alpha'$.
\end{compactitem}
Furthermore, if $\astore(u) \in \dom{\aheap'} \setminus
\{\astore(\fvar_i) \mid 1 \leq i \leq \nfv \}$ then $u \in
\allocf{\alpha'}$. 
\end{lemma}
\optproof{\subsection{Proof of Lemma \ref{lem:form2deco}}}
{
  Note that by hypothesis, $L$ cannot contain $\bots$, since
  $\astore(\bot) = \bots$ and $\img{\astore} \cap L = \emptyset$. The
  proof is by induction on the pair $(\len{\aheap},\size{\alpha})$,
  using the lexicographic order. We distinguish several
  cases: \begin{compactitem}
  \item{If $\alpha$ is of the form $x \iseq y$ or $x \not \iseq y$,
    then $\aheap = \emptyset$, and $\alpha$ is a \decoration of
    itself, since it contains no predicate symbol. Since
    $(\astore,\aheap) \expands{\id} (\astore,\aheap)$, we may thus set
    $\alpha' \isdef \alpha$, $\gamma \isdef \id$ and $\aheap'\isdef
    \aheap$.}
  \item{If $\alpha'$ is of the form $\alpha_1' \vee \alpha_2'$ then
    the proof follows immediately from the induction hypothesis.}
  \item{If $\alpha$ is of the form $\alpha_1 * \alpha_2$, then let
    $L_1,L_2$ be two disjoint infinite subsets of $L$. Since
    $(\astore,\aheap) \models_{\asid} \alpha_1 * \alpha_2$,
 %   \comment[ri]{instead of by definition}, 
 there exist disjoint heaps
    $\aheap_1$, $\aheap_2$ such that $\aheap = \aheap_1 \dunion
    \aheap_2$ and $(\astore,\aheap_i) \models_{\asid} \alpha_i$.  By
    the induction hypothesis, for $i = 1,2$, there exists a
    {\decoration} $\alpha_i'$ of $\alpha_i$, a heap $\aheap_i'$ and a
    mapping $\gamma_i: \Loc \rightarrow \Loc$ such that:
    $(\astore,\aheap_i')\expands{\gamma_i} (\astore,\aheap_i)$ ; $\ell
    \not \in L_i \Rightarrow \gamma_i(\ell) = \ell$; $\locs{\aheap'_i}
    \setminus \img{\astore} \subseteq L_i$; $\dom{\bpart{\aheap'}}
    \subseteq L_i$; $(\astore,\aheap_i')\models_{\leftsid} \alpha_i'$
    and if $\astore(u) \in \dom{\aheap'_i} \setminus
    \{\astore(\fvar_j) \mid 1 \leq j \leq \nfv \}$ then $u \in
    \allocf{\alpha'_i}$.  Let $\alpha' \isdef \alpha_1' * \alpha_2'$
    and consider the function
    \[
    \gamma: \ell \mapsto
    \begin{cases}
      \gamma_1(\ell) & \text{if $\ell \in L_1$} \\
      \gamma_2(\ell) & \text{if $\ell \in L_2$}\\
      \ell & \text{otherwise}
    \end{cases}
    \]
    Since $L_1 \cap L_2 = \emptyset$, this function is well-defined,
    and since $L_1 \cup L_2 \subseteq L$, if $\ell \not \in L$ then
    $\gamma(\ell) = \ell$. Assume that $\dom{\aheap_1'} \cap
    \dom{\aheap_2'}$ contains an element $\ell$. Then by the induction
    hypothesis, for $i = 1, 2$, $\ell \in \img{\astore} \cup L_i$; and
    since $L_1 \cap L_2 = \emptyset$, we deduce that $\ell \in
    \img{\astore}$, so that $\ell \not \in L$.  Since
    $\dom{\bpart{\aheap_i'}} \subseteq L_i$, necessarily $\ell \in
    \dom{\mpart{\aheap_i'}}$ and
    $(\astore,\aheap_i')\expands{\gamma_i} (\astore,\aheap_i)$, by
    Condition \ref{it:decor:dom} of Definition \ref{def:decor_heap},
    $\gamma_i(\ell) \in \dom{\aheap_i}$.  Since $\ell \not \in L$ we
    have $\gamma_1(\ell) = \gamma_2(\ell) = \ell$, and we deduce that
    $\ell \in \dom{\aheap_1} \cap \dom{\aheap_2}$, which contradicts
    the fact that $\aheap_1$ and $\aheap_2$ are disjoint.
    Consequently, $\aheap_1'$ and $\aheap_2'$ are disjoint and we may
    define $\aheap' \isdef \aheap_1' \dunion \aheap_2'$. Since
    $(\astore,\aheap_i') \models_{\leftsid} \alpha_i'$, for both
    $i=1,2$, we have $(\astore,\aheap') \models_{\leftsid} \alpha'$.
	
    Let $\ell \in \locs{\aheap_i'}$ for $i = 1,2$. If $\ell \in
    \img{\astore}$, then by hypothesis $\ell \notin L$ and by
    construction, $\gamma(\ell) = \ell = \gamma_i(\ell)$. Otherwise,
    $\ell \in \locs{\aheap_i'} \setminus \img{\astore}$, thus by the
    induction hypothesis $\ell \in L_i$ and by construction,
    $\gamma(\ell) = \gamma_i(\ell)$. We deduce by Lemma
    \ref{prop:expand-dunion} that $(\astore,\aheap') \expands{\gamma}
    (\astore,\aheap)$. Furthermore, still by the induction hypothesis
    we have: \[\locs{\aheap'}\setminus \img{\astore} \subseteq
    (\locs{\aheap_1'} \setminus \img{\astore}) \cup (\locs{\aheap_2'}
    \setminus \img{\astore}) \subseteq L_1 \cup L_2 \subseteq L.\] We
    also have $\bpart{\aheap'} = \bpart{\aheap_1'} \dunion
    \bpart{\aheap_2'}$, thus $\dom{\bpart{\aheap'}} =
    \dom{\bpart{\aheap_1'}} \cup \dom{\bpart{\aheap_2'}} \subseteq L_1
    \cup L_2 \subseteq L$. Finally, if $\astore(x)\in
    \dom{\aheap'}\setminus \{\astore(\fvar_i) \mid 1 \leq i \leq \nfv
    \}$ then necessarily $\astore(x) \in \dom{\aheap_i'}$ for some $i
    = 1,2$, so that $x\in \allocf{\alpha_i'}$ (by the induction
    hypothesis, last point of the lemma) and therefore $x\in
    \allocf{\alpha'}$.  }
  \item{Assume that $\alpha$ is of the form $p(u_1,\dots,u_n)$ and
    that $\dom{\astore} = \fv{\alpha} \cup \{
    \fvar_1,\dots,\fvar_\nfv, \bot\}$. Then, $\asid$ contains a rule
    $p(x_1,\dots,x_n) \Leftarrow x_1 \mapsto (y_1,\dots,y_\rank) *
    \genrest$ such that $(\astorep,\aheap) \models_{\asid} u_1 \mapsto
    (y_1,\dots,y_\rank)\theta * \genrest\theta$, where $\theta =
    \substinterv{x_i}{u_i}{i}{\interv{1}{n}}$ and $\astorep$ is an
    extension of $\astore$.  Let $\{z_1,\dots,z_\maxv\} \isdef
    (\fv{\genrest} \cup \{ y_1,\dots,y_\rank\}) \setminus \{
    x_1,\dots,x_n\}$, be the set of existential variables of the above
    rule. We have $\dom{\astore_e} = \dom{\astore} \cup \{
    z_1,\dots,z_\maxv \}$. Consider the substitution $\sigma$ such
    that: $\dom{\sigma} \subseteq \{ z_1,\dots,z_\maxv \}$ and
    $\sigma(z_i) = z$ iff $z$ is the first variable in $u_1,\dots,u_n,
    \fvar_1,\dots,\fvar_\nfv, z_1,\dots,z_{i-1}$ such that
    $\astorep(z_i) = \astorep(z)$ ($\sigma(z_i)$ is undefined in there
    is no such variable). By construction, if $x$ is a variable
    occurring in $\genrest\sigma$, then $x\notin \dom{\sigma}$.  Let
    $\astorei$ be the restriction of $\astorep$ to $\dom{\astorep}
    \setminus \dom{\sigma}$ and we show that $\astorei$ is
    quasi-injective. Assume that $\astorei(x) = \astorei(x')$ for
    distinct variables $x,x' \in \dom{\astorei}$ with $\{ x,x' \} \not
    \subseteq \{ \fvar_i \mid i \in \interv{1}{\nfv}\}$. Since
    $\astorei$ is a restriction of $\astorep$, we have $\astorep(x) =
    \astorep(x')$. We deduce that $x$ and $x'$ both occur in the
    sequence
    $u_1,\dots,u_n,\fvar_1,\dots,\fvar_\nfv,z_1,\dots,z_\maxv$, and we
    assume w.l.o.g. that $x$ occurs before $x'$ in this sequence.  If
    $x,x' \in \set{u_1,\ldots, u_n,\fvar_1,\dots,\fvar_\nfv}$, then
    since $\astorep$ is an extension of $\astore$, we would have
    \(\astore(x)\ =\ \astorep(x)\ =\ \astorep(x')\ =\ \astore(x')\),
    so that $x = x'$, because by hypothesis $\astore$ is
    quasi-injective. Thus, one of the variables $x,x'$ is in $\{
    z_1,\dots,z_\maxv \}$. Since $x$ occurs before $x'$ in the
    sequence
    $u_1,\dots,u_n,\fvar_1,\dots,\fvar_\nfv,z_1,\dots,z_\maxv$, we
    deduce that $x' \in \{ z_1,\dots,z_\maxv \}$.  By definition of
    $\sigma$, this entails that $x'\sigma = x\sigma \neq x'$, hence
    $x' \in \dom{\sigma}$ and $x' \not \in \dom{\astorei}$, which
    contradicts our assumption.

    Let $\aheap_1$ be the restriction of $\aheap$ to $\dom{\aheap}
    \setminus \{ \astore(u_1) \}$, so that $(\astorep,\aheap_1)
    \models_{\asid} \genrest\theta$.  Then by construction,
    $(\astorei,\aheap_1)\models_{\asid} \genrest\sigma\theta$.  Let
    $L_1 \isdef L \setminus \img{\astorei}$.  Since $\locs{\aheap_1}
    \subseteq \locs{\aheap}$ and $L_1\subseteq L$, we have
    $(\img{\astorei}\cup \locs{\aheap_1}) \cap L_1 = \emptyset$.
    Thus, by the induction hypothesis, there exists a \decoration
    $\genrest'$ of $\genrest\sigma\theta$, a mapping $\gamma_1: \Loc
    \rightarrow \Loc$ satisfying $\ell\not \in L_1 \Rightarrow
    \gamma_1(\ell) = \ell$ and a heap $\aheap_1'$ satisfying
    $\locs{\aheap'_1} \setminus \img{\astorei} \subseteq L_1$ and
    $\dom{\bpart{\aheap_1'}} \subseteq L_1$, such that
    $(\astorei,\aheap_1') \expands{\gamma_1} (\astorei,\aheap_1)$,
    $(\astorei,\aheap_1') \models_{\leftsid} \genrest'$
     and for all variables $u$, if $\astorei(u)
    \in \dom{\aheap_1'} \setminus \setof{\astorei(\fvar_i)}{1\leq
      i\leq\nfv}$, then $u\in \allocf{\genrest'}$.
 
Let $E \isdef (\img{\astorei} \cup \locs{\aheap'_1}) \setminus
\img{\astore}$ and consider a bijection $\eta: \Loc\rightarrow \Loc$
such that: \begin{compactitem}
\item if $\ell \in E$ then $\eta(\ell) \in L\setminus E$ and
  $\eta(\eta(\ell)) = \ell$;
\item if $\ell \in \Loc \setminus(E\cup \eta(E))$ then $\eta(\ell) =
  \ell$.
\end{compactitem} 
Such a bijection necessarily exists because $E$ is finite and $L$ is
infinite. Let $\ell \in \img{\astore}$, so that $\ell \notin E$, and
assume $\ell \in \eta(E)$. Then $\inv{\eta}(\ell) \in E$, hence $\ell
\in L\setminus E$. But by the hypotheses of the lemma, $(\img{\astore}
\cup \locs{\aheap}) \cap L = \emptyset$, so this case is
impossible. We deduce that $\ell \in \Loc \setminus (E\cup \eta(E))$
and that $\eta(\ell) = \ell$.  Thus, in particular, $\eta(\bots) =
\bots$.  Consider the mapping $\gamma \isdef \gamma_1 \circ
\inv{\eta}$, the heap $\aheap_1'' \isdef \apl{\eta}{\aheap_1'}$ and
the store $\astorei' \isdef \apl{\eta}{\astorei}$.  By Lemma
\ref{prop:expand-bij} $(\astorei',\aheap_1'')
\expands{\gamma}(\astorei',\aheap_1)$.  By Lemma \ref{lem:gamma_inj},
since $(\astorei,\aheap_1') \models_{\leftsid} \genrest'$, we deduce
that $(\astorei',\aheap_1'') \models_{\leftsid} \genrest'$.  We have
$\dom{\astore} \subseteq \dom{\astorep}\setminus \dom{\sigma} =
\dom{\astorei} = \dom{\astorei'}$, hence the restriction of
$\astorei'$ to $\dom{\astore}$ is well-defined, and if $x\in
\dom{\astore}$, then
\[\astorei'(x)\ = \ \eta(\astorei(x))\ =\ \eta(\astorep(x))\ =\ \eta(\astore(x))\ =\ \astore(x).\] This shows that the restriction of $\astorei'$ to $\dom{\astore}$ coincides with $\astore$.

Let $j \in \interv{1}{\maxv}$ such that $z_j\not \in \dom{\sigma}$.
By definition we have $\astorei(z_j) \not \in \img{\astore}$, thus
$\astorei(z_j) \in E$ and $\astorei'(z_j) = \eta(\astorei(z_j)) \in
L$.  Let $I$ be the set of indices $j \in \interv{1}{\maxv}$ such that
$z_j \not \in \dom{\sigma}$ and $\astorei(z_j) \not \in
\dom{\aheap_1'}$; for all $j\in I$, we therefore have $\astorei'(z_j)
\in L$ (indeed, by definition of $\sigma$, $\astorei(z_j) =
\astorep(z_j) \not \in \img{\astore}$, thus $\astorei(z_j) \in E$, and
by definition of $\eta$, $\eta(\astorei(z_j)) \in L$, hence
$\astorei'(z_j)\in L$). Consider the set:
\[X\ \isdef\ \{ 1\} \cup \{ i \in \interv{1}{n} \mid u_i \in \allocf{\genrest'} 
\wedge \forall j \in \interv{1}{\nfv},\,\astore(u_i) \not=
\astore(\fvar_j) \},\] and let $\alpha' \isdef
\decorated{p}{X}(u_1,\dots,u_n,\mvec{\fvar})$. 
By definition,
$\decor{\asid}$ contains a rule $(\mp)$ of the form
 \[\decorated{p}{X}(x_1,\dots,x_n,\mvec{\fvar}) \Leftarrow 
 x_1 \mapsto (y_1,\dots,y_\rank,\mvec{\fvar},z_1,\dots,z_\maxv)\sigma
 * \genrest'' * \bigast_{i \in I} \botp(z_i), \] where
 $\genrest''\theta = \genrest'$. We define the following heaps:
 \[ \begin{array}{lll}
   \aheap_2' & \isdef & 
   \{ \tuple{\astorei'(u_1),(\astorei'(y_1\sigma\theta),\dots,\astorei'(y_\rank\sigma\theta),\astorei'(\mvec{\fvar}),
     \astorei'(z_1\sigma\theta),\dots,\astorei'(z_\maxv\sigma\theta))} \}, \\
   \aheap_3' & \isdef & \{ \tuple{\astorei'(z_j),\mvec{\bots}} \mid j \in I \}.
 \end{array}\]
 By definition of $I$, it cannot be the case that $\astorei'(u_1) =
 \astorei'(z_j)$ for $j\in I$, because otherwise we would have
 $\astore(u_1) = \astorei(z_j)$ and $z_j \in \dom{\sigma}$, hence
 $\dom{\aheap_2'} \cap \dom{\aheap_3'} = \emptyset$.  We show that
 $(\dom{\aheap_2'} \cup \dom{\aheap_3'}) \cap \dom{\aheap_1''} =
 \emptyset$.  First let $j\in I$, and assume $\astorei'(z_j) \in
 \dom{\aheap_1''}$. Then since $\eta$ is a bijection, necessarily,
 $\astorei(z_j) \in \dom{\aheap_1'}$, which is impossible by
 definition of $I$. Now assume that $\astorei'(u_1) \in
 \dom{\aheap_1''}$, so that $\astorei(u_1) \in \dom{\aheap_1'}$. Then
 by definition of $L_1$ we have $\astorei(u_1) \notin L_1$, and
 $\gamma_1(\astorei(u_1)) = \astorei(u_1)$. Since
 $\dom{\bpart{\aheap_1'}} \subseteq L_1$, necessarily $\astorei(u_1)
 \in \dom{\mpart{\aheap_1'}}$, hence $\gamma_1(\astorei(u_1)) =
 \astorei(u_1) = \astore(u_1) \in \dom{\aheap_1}$, which is impossible
 by definition of $\aheap_1$.  This shows that the domains of
 $\aheap_1''$, $\aheap_2'$ and $\aheap_3'$ are pairwise disjoint, that
 $\aheap' \isdef \aheap_1'' \dunion \aheap_2' \dunion \aheap_3'$ is
 well-defined, and by construction,
\[(\astorei',\aheap') \models_{\leftsid} u_1 \mapsto (y_1,\dots,y_\rank,\mvec{\fvar},z_1,\dots,z_\maxv)\sigma\theta 
 * \genrest''\theta * \bigast_{i \in I} \botp(z_i)\theta.\] We show
 that $(\astore,\aheap') \expands{\gamma} (\astore,\aheap)$, with
 $\mpart{\aheap'} \isdef \mpart{\aheap_1''} \dunion \aheap_2'$ and
 $\bpart{\aheap'} \isdef \bpart{\aheap_1''} \dunion \aheap_3'$.  Note
 that if $\ell\in \dom{\aheap_2'}$ then necessarily $\ell =
 \astorei'(u_1) = \astore(u_1)\in \img{\astore}$, so that $\ell \notin
 L_1$ and by definition of $\eta$ and $\gamma_1$, $\gamma(\ell) =
 \gamma_1\circ\inv{\eta}(\ell) = \gamma_1(\ell) = \ell$.  We check the
 four points of Definition \ref{def:decor_heap} below: \begin{compactenum}
 \item Let $\ell_1, \ell_2$ be locations in $\dom{\mpart{\aheap'}}$
   such that $\gamma(\ell_1) = \gamma(\ell_2)$, and assume that
   $\ell_1 \in \dom{\mpart{\aheap_1''}}$. Then $\gamma(\ell_1) \in
   \dom{\aheap_1}$ because $(\astorei',\aheap_1'')
   \expands{\gamma}(\astorei',\aheap_1)$.  If $\ell_2\in
   \dom{\aheap_2'}$ then $\gamma(\ell_2) = \ell_2$. By definition of
   $\aheap_1$ we cannot have $\ell_2\in \dom{\aheap_1}$, and
   therefore, $\ell_1\neq \ell_2$. Otherwise, $\ell_2\in
   \dom{\mpart{\aheap_1''}}$ and for $i=1,2$ we have $\gamma(\ell_i) =
   \gamma_1(\inv{\eta}(\ell_i))$ and $\inv{\eta}(\ell_i) \in
   \dom{\aheap_1'}$. Since $(\astorei, \aheap_1') \expands{\gamma_1}
   (\astorei, \aheap_1)$, we deduce that $\inv{\eta}(\ell_1) =
   \inv{\eta}(\ell_2)$ and because $\eta$ is a bijection, $\ell_1 =
   \ell_2$. The proof is symmetric if $\ell_2 \in
   \dom{\mpart{\aheap_1''}}$.  Finally, if $\ell_1,\ell_2 \in
   \dom{\aheap_2'}$, then $\ell_1 = \ell_2$ since $\dom{\aheap_2'}$ is
   a singleton.
	
%	 Since $\apl{\eta}{\aheap_1^1}$ and $\aheap_2'$ are disjoint
%	and $\dom{\aheap_2'}$ contains only one element, we have
%	$\ell_1,\ell_2 \in \dom{\apl{\eta}{\aheap_1^1}}$.  Then, since
%	$\dom{\aheap_3'} \cap \dom{\apl{\eta}{\aheap_1^1}} =
%	\emptyset$ we deduce that $\gamma(\ell_i) =
%	\gamma_1'(\ell_i)$, thus $\gamma_1'(\ell_1)=\gamma_2'(\ell_2)$
%	and $\ell_1 =\ell_2$.
	
 \item We have $\gamma(\dom{\mpart{\aheap_1''}}) = \dom{\aheap_1}$ and
   since $\gamma(\astorei'(u_1)) = \astorei'(u_1) = \astore(u_1)$, we
   deduce that $\gamma(\dom{\apl{\eta}{\mpart{\aheap_1''}}} \cup
   \dom{\aheap_2'}) = \dom{\aheap}$.
 \item Let $\ell \in \dom{\mpart{\aheap_1''}} \cup
   \dom{\aheap_2'}$. If $\ell = \astorei'(u_1)$, then by construction
   we have
  \begin{eqnarray*}
    \aheap'(\astorei'(u_1)) & = & 
    (\astorei'(y_1\sigma\theta),\dots,\astorei'(y_\rank\sigma\theta),\astorei'(\mvec{\fvar}),
    \astorei'(z_1\sigma\theta),\dots,\astorei'(z_\maxv\sigma\theta)) \\
    & = & (\astorei'(y_1\sigma\theta),\dots,\astorei'(y_\rank\sigma\theta),\astore(\mvec{\fvar}),
    \astorei'(z_1\sigma\theta),\dots,\astorei'(z_\maxv\sigma\theta)),
  \end{eqnarray*}
  where the second line follows from the fact that $\astorei'$
  coincides with $\astore$ on $\dom{\astore}$.  Note that, using the
  fact that $\img{\astorei} \cap L_1 = \emptyset$ and by definition of
  $\gamma$ and $\astorei'$, the following equalities hold:
  \begin{eqnarray*}
    \gamma\left(\astorei'(y_1\sigma\theta),\dots,\astorei'(y_\rank\sigma\theta)\right) & = & 
    (\gamma_1(\astorei(y_1\sigma\theta)), \ldots, \gamma_1(\astorei(y_\rank\sigma\theta)))\\
    & = & (\astorei(y_1\sigma\theta), \ldots, \astorei(y_\rank\sigma\theta)).
  \end{eqnarray*}
 		 	
% \comment[me]{to recheck} where the second equality is obtained using
%the fact that by definition of $\gamma$, for all $j \in
%\interv{1}{k}$, we have $\gamma(\astorei(y_j\sigma\theta)) =
%\astorei'(y_j\sigma\theta)$.  \comment[me]{In order to prove that
%$(\astore,\aheap') \expands{\gamma} (\astore,\aheap)$, don't we need
%to be reasoning on $\astore$ instead of $\astorei'$?}
%\comment[np]{yes, actually, the relation $(\astore,\aheap')
%\expands{\gamma} (\astore,\aheap)$ does not depend on $\astore$,
%except for the interpretation of the variables
%$\fvar_1,\dots,\fvar_\nfv$ and they are the same in all stores (as
%shown above). I added a line above. }
We also have: 
\begin{eqnarray*}
  \aheap(\astorep(u_1)) & = & (\astorep(y_1\theta), \ldots, \astorep(y_\rank\theta))\ \text{ and}\\
  \aheap(\astorei(u_1)) & = & (\astorei(y_1\sigma\theta), \ldots, \astorei(y_\rank\sigma\theta)), 
\end{eqnarray*}
where the second equation is a consequence of the definitions of
$\sigma$ and $\astorei$ respectively. This proves that:
\begin{eqnarray*}
  \aheap(\gamma(\astorei'(u_1)))\ = \ \aheap(\astorei(u_1)) & = & 
  \left(\astorei(y_1\sigma\theta), \ldots, \astorei(y_\rank)\sigma\theta\right) \\
  & = & \gamma\left(\astorei'(y_1\sigma\theta),\dots,\astorei'(y_\rank\sigma\theta)\right),
\end{eqnarray*}
hence that $\astorei'(u_1)$ satisfies Condition \ref{it:decor:img} of
Definition \ref{def:decor_heap}.  Now if $\ell \in
\dom{\mpart{\aheap_1''}}$, then it is straightforward to verify that
$\ell$ satisfies Condition \ref{it:decor:img} of Definition
\ref{def:decor_heap}, using the fact that
 $(\astorei', \aheap_1'') \expands{\gamma} (\astorei', \aheap_1)$.
\item If $\ell \in \dom{\aheap_3'}$, then $\ell =
  \astorei'(z_j\theta)$ for some $j\in I$ and it is simple to verify
  that Condition \ref{it:decor:link} of Definition
  \ref{def:decor_heap} is verified, setting $\gcon{\ell}{\aheap'}
  \isdef \astorei'(u_1)$.  If $\ell \in \dom{\bpart{\aheap_1''}}$
  then, using the fact that $(\astorei', \aheap_1'') \expands{\gamma}
  (\astorei', \aheap_1)$, we deduce that Condition \ref{it:decor:link}
  of Definition \ref{def:decor_heap} is verified.
\end{compactenum}
We prove that $\dom{\bpart{\aheap'}} \subseteq L$. Let $\ell \in
\dom{\bpart{\aheap'}}$, and first assume that $\ell \in
\dom{\bpart{\aheap_1''}}$, so that $\ell = \eta(\ell')$ for $\ell'\in
\dom{\bpart{\aheap_1'}}$.  By the induction hypothesis we have
$\dom{\bpart{\aheap'_1}} \subseteq L_1$, thus $\ell' \in L_1$. If
$\ell' \in \img{\astore}$, then $\ell'\not \in E$ (by definition of
$E$), and $\ell' \not\in L$ (by the hypothesis of the lemma), hence by
definition of $\eta$ we have $\eta(\ell') = \ell' = \ell \in L_1
\subseteq L$. Otherwise $\ell'\in E$ and by construction, $\eta(\ell')
\in L\setminus E \subseteq L$. Now assume that $\ell \in
\dom{\aheap_3'}$. Then $\ell = \astorei'(z_j)$ for some $j\in I$, and
since we have shown that $\astorei'(z_j)\in L$, for every $j \in I$,
we have $\ell \in L$.

We now show that $\locs{\aheap'} \setminus \img{\astore} \subseteq L$.
Since $\dom{\bpart{\aheap'}} \subseteq L$ and $\locs{\bpart{\aheap'}}
= \dom{\bpart{\aheap'}} \cup \{ \bots \}$, we deduce that
$\locs{\bpart{\aheap'}}\setminus \img{\astore} \subseteq L$, because
$\bots \in\img{\astore}$.  Now let $\ell \in \locs{\mpart{\aheap'}}
\setminus \img{\astore}$.  If $\ell \in \locs{\aheap_2'}$ then by
definition of $\aheap_2'$ and since $\ell \not \in \img{\astore}$, we
must have $\ell = \astorei'(z_j\sigma\theta) = \astorei'(z_j)$ for
some $j\in I$, hence $\ell \in L$. Otherwise $\ell \in
\locs{\mpart{\aheap_1''}}$, and $\ell = \eta(\ell')$ for some $\ell'
\in \locs{\mpart{\aheap_1'}}$.  If $\ell' \in \img{\astore}$ then
$\eta(\ell') = \ell \in \img{\astore}$, which contradicts our
assumption.  Thus $\ell' \not \in \img{\astore}$, hence $\ell' \in E$,
and by definition of $\eta$, $\eta(\ell') = \ell \in L$.

There remains to prove that Rule $(\mp)$ is \welldefined; this entails
that it occurs in $\leftsid$, hence that $(\astore,\aheap')
\models_{\leftsid} \alpha'$.  We first check that Condition
\ref{welldefined:propagate} of Definition \ref{def:welldefined}
holds. By construction we have $1\in X$, hence $x_1 \in
\allocf{\decorated{p}{X}(x_1,\dots,x_n)}$, and we also have $x_1 \in
\allocf{x_1 \mapsto
  (y_1,\dots,y_\rank,\mvec{\fvar},z_1,\dots,z_\maxv)\sigma}$.  Now
assume that $x_i \in \allocf{\decorated{p}{X}(x_1,\dots,x_n)}$ and
that $i \not = 1$.  By definition of $X$, this entails that $x_i\theta
\in \allocf{\genrest'}$.  Since $\genrest' = \genrest''\theta$, we
deduce that $x_i\in \allocf{\genrest''}$.  Next, we check that
Condition \ref{welldefined:established} of Definition
\ref{def:welldefined} holds.  Let $z$ be a variable occurring on the
right-hand side of rule $(\mp)$ but not on its left-hand side.  Then
$z = z_j$, for some $j$ with $z_j \not \in \dom{\sigma}$ (indeed,
$z_1,\dots,z_\maxv$ are the only existential variables, and if $z\in
\dom{\sigma}$ then by definition of $\sigma$, we have $\sigma(z') \not
= z$ for every variable $z'$, thus $z$ cannot occur in
$\genrest\sigma$, hence in $\genrest'$) and by definition of $\sigma$,
we have $\astorei(z_j) = \astorep(z_j) \not \in \{
\astore(\fvar_1),\dots,\astore(\fvar_\nfv) \}$.  If $j \in I$ then
$z_j \in \allocf{\bigast_{i \in I} \botp(z_i)}$. Otherwise, by
definition of $I$, we must have $\astorei(z_j) \in \dom{\aheap_1'}$,
and since $\img{\astorei} \cap L_1 = \emptyset$, necessarily,
$\astorei(z_j) = \gamma_1(\astorei(z_j))$.  By the induction
hypothesis, since $\astorei(z_j) \not \in \{
\astore(\fvar_1),\dots,\astore(\fvar_\nfv) \}$ and $\astorei(\fvar_i)
= \astore(\fvar_i)$ for $i \in \interv{1}{\nfv}$, we deduce that $z_j
\in \allocf{\genrest'}$; and since $z_j \not \in \dom{\theta}$, we
must have $z_j \in \allocf{\genrest''}$.
 
We finally show that if $\astore(u) \in \dom{\aheap'} \setminus
\{\astore(\fvar_i) \mid 1 \leq i \leq \nfv \}$ then $u\in
\allocf{\alpha'}$. Consider such a variable $u$. Assume $\astore(u)\in
\dom{\aheap_3'}$. Then $\astore(u)$ is of the form $\astorei'(z_j)$
for some $j\in I$, hence $\astore(u) \in L$, since we have shown that
$\astorei'(z_j) \in L$, for every $j \in I$. But $L \cap \img{\astore}
= \emptyset$, so this case is impossible. We deduce that $\astore(u)
\in \dom{\aheap_1''} \cup \{ \astore(u_1) \}$.  If $\astore(u) =
\astore(u_1)$, then since $\astore$ is quasi-injective we deduce that
$u = u_1$ thus $u\in \allocf{\alpha'}$, because by construction, $1
\in X$.  Otherwise, we have $\astore(u) \in \dom{\aheap_1''}$ and
since $\eta(\ell) = \ell$ for all $\ell \in \img{\astore}$,
necessarily $\astore(u) \in \dom{\aheap_1'}$. By the induction
hypothesis, we deduce that $u \in \allocf{\genrest'}$, hence there
exists $x\in \allocf{\genrest''}$ $(1\leq j \leq n$) such that $u =
x\theta$.  By definition of $\theta$ (and assuming by renaming that
$\fv{\genrest} \cap \fv{\alpha} = \emptyset$), necessarily, $x = x_j$,
for some $j \in \interv{1}{n}$. Then by definition of $X$ we have
$j\in X$, thus $u = x_j\theta \in \allocf{\alpha'}$.  }
  \item{Assume that $\alpha$ is of the form $p(u_1,\dots,u_n)$ and
    that $\dom{\astore} \not = \fv{\alpha} \cup \{
    \fvar_1,\dots,\fvar_\nfv, \bots\}$. Note that we have necessarily
    $\dom{\astore} \supseteq \fv{\alpha} \cup \{
    \fvar_1,\dots,\fvar_\nfv, \bots\}$.  Let $\astore'$ be the
    restriction of $\astore$ to $\fv{\alpha} \cup \{
    \fvar_1,\dots,\fvar_\nfv, \bots\}$.  It is clear that
    $(\astore',\aheap) \models \alpha$ and that $\astore'$ fulfills
    all the hypotheses of the lemma.  Thus, by the previous item,
    there exists $\alpha'$, $\aheap'$ and $\gamma: \Loc
    \rightarrow\Loc$ such that $(\astore',\aheap') \expands{\gamma}
    (\astore',\aheap)$, if $\ell \not \in L$ then $\gamma(\ell) =
    \ell$, $\locs{\aheap'} \setminus \img{\astore'} \subseteq L$,
    $(\astore',\aheap') \models_{\leftsid} \alpha'$ and $\astore(u)
    \in \dom{\aheap'} \setminus \{\astore(\fvar_i) \mid 1 \leq i \leq
    \nfv \}\Rightarrow u \in \allocf{\alpha'}$.  It is clear that we
    have $(\astore,\aheap') \expands{\gamma} (\astore,\aheap)$, and
    $(\astore,\aheap') \models_{\leftsid} \alpha'$.  Furthermore,
    since $\img{\astore'} \subseteq \img{\astore}$, we also have
    $\locs{\aheap'} \setminus \img{\astore} \subseteq L$.  Finally, if
    $\astore(u) \in \dom{\aheap'} \setminus \{\astore(\fvar_i) \mid 1
    \leq i \leq \nfv \}$, then, since $\astore(u) \not \in L$, and
    $\imgh{\aheap'} \setminus \img{\astore'} \subseteq L$ we must have
    $\astore(u) \in \img{\astore'}$, thus $u \in \allocf{\alpha'}$ by
    the previous item. \qed}
\end{compactitem}
}

 \subsection{\capitalisewords{Transforming Entailments}}

 \newcommand{\nf}[1]{#1\downarrow}
 
We define $\rsid \isdef \leftsid \cup \rightsid$. We show that the
instance $\phi \vdash_{\asid} \psi$ of the safe entailment problem can
be solved by considering an entailment problem on $\rsid$ involving
the elements of $\decor{\phi}$ (see Definition
\ref{def:deco}). \comment[ri]{added} Note that the rules from
$\leftsid$ are progressing, connected and established, by Lemma
\ref{cor:pce}, whereas the rules from $\rightsid$ are progressing and
connected, by Definition \ref{def:right_rule}. Hence, each entailment
problem $\phi' \vdash_{\rsid} \rightdecor{\psi}$, where $\phi' \in
\decor{\phi}$, is progressing, connected and left-established.

\comment[ri]{changed the corollary into a lemma}
\begin{lemma}\label{cor:safe-equiv}
  $\phi \models_{\asid} \psi$ if and only if $\bigvee_{\phi' \in
    \decor{\phi}} \phi' \models_{\rsid} \rightdecor{\psi}$.
\end{lemma}
\begin{proof}
  ``$\Rightarrow$'' Assume that $\phi \models_{\asid} \psi$ and let
  $\phi'\in \decor{\phi}$ be a formula, $(\astore,\aheap')$ be an
  $\rsid$-model of $\phi'$ and $\aheap \isdef \trunc{\aheap'}$. By
  construction, $(\astore,\aheap')$ is an $\leftsid$-model of
  $\phi'$. By definition of $\decor{\phi}$, $\phi'$ is a \decoration
  of $\phi$.  Let $D_2 \isdef \{ \ell \in \dom{\aheap'} \mid
  \aheap'(\ell) = \mvec{\bots} \}$, $D_1 \isdef \dom{\aheap'}
  \setminus D_2$, and consider a location $\ell \in \dom{\aheap'}$.
  By definition, $\ell$ must be allocated by some rule in $\leftsid$.
  If $\ell$ is allocated by a rule of the form given in Definition
  \ref{def:deco_rule}, then necessarily $\aheap'(\ell)$ is of the form
  $(\ell_1,\dots,\ell_\rank,\astore(\fvar),\ell_1',\dots,\ell'_\maxv)$
  and $\ell \in D_1$. Otherwise, $\ell$ is allocated by the predicate
  $\botp$ and we must have $\ell \in D_2$ by definition of the only
  rule for $\botp$. Since this predicate must occur within a rule of
  the form given in Definition \ref{def:deco_rule}, $\ell$ necessarily
  occurs in the $\maxv$ last components of the image of a location in
  $D_1$, hence admits a \gconnection in $\aheap'$.  Consequently, by
  Lemma \ref{prop:trunc} $(\astore, \aheap') \expands{\id} (\astore,
  \aheap)$, and by Lemma \ref{lem:deco2form}, $(\astore,\aheap)
  \models_{\asid} \phi$.  Thus $(\astore,\aheap) \models_{\asid}
  \psi$, and by Lemma \ref{lem:right_equiv}, $(\astore,\aheap')
  \models_{\rightsid} \rightdecor{\psi}$, thus $(\astore,\aheap')
  \models_{\rsid} \rightdecor{\psi}$.
  
  ``$\Leftarrow$'' Assume that $\bigvee_{\phi' \in \decor{\phi}} \phi'
  \models_{\rsid} \rightdecor{\psi}$ and let $(\astore,\aheap)$ be a
  $\asid$-model of $\phi$. Since the truth values of $\phi$ and $\psi$ depend
  only on the variables in $\fv{\phi} \cup \fv{\psi}$, we may assume,
  w.l.o.g., that $\astore$ is quasi-injective.  Consider an infinite
  set $L\subseteq \Loc$ such that $(\img{\astore} \cup \locs{\aheap})
  \cap L = \emptyset$.  By Lemma \ref{lem:form2deco}, there exist a
  heap $\aheap'$, a mapping $\gamma:\Loc \rightarrow \Loc$ and a
  \decoration $\phi'$ of $\phi$ such that $\gamma(\ell) = \ell$ for
  all $\ell \notin L$, $(\astore,\aheap') \expands{\gamma}
  (\astore,\aheap)$ and $(\astore,\aheap')\models \phi'$. Since
  $\img{\astore} \cap L =\emptyset$, we also have
  $\apl{\gamma}{\astore} = \astore$.  Then $(\astore,\aheap') \models
  \rightdecor{\psi}$.  Let $\aheap_1 \isdef \trunc{\aheap'}$. Since
  $(\astore,\aheap') \expands{\gamma} (\astore,\aheap)$, by Lemma
  \ref{prop:exp-trunc} we have $(\astore,\aheap') \expands{\id}
  (\astore,\aheap_1)$, and by Lemma \ref{lem:right_equiv},
  $(\astore,\aheap_1) \models \psi$.  By Lemma \ref{prop:exp-trunc} we
  have $\aheap = \gamma(\aheap_1)$. Since $\psi$ is \restricted
  w.r.t.\ $\{ \fvar_1,\dots,\fvar_n\}$, we deduce by Lemma
  \ref{lem:gamma} that $(\astore,\aheap) \models \psi$. \qed
   \end{proof}

This leads to the main result of this paper: 

\begin{theorem}\label{thm:safe-complexity}
  The safe entailment problem is 2EXPTIME-complete. 
\end{theorem}
\begin{proof}
  The 2EXPTIME-hard lower bound follows from \cite[Theorem
    32]{EIP21a}, as the class of progressing, \connected and
  \restricted entailment problems is a subset of the safe entailment
  class. For the 2EXPTIME membership, Lemma \ref{cor:safe-equiv}
  describes a many-one reduction to the progressing, connected and
  established class, shown to be in 2EXPTIME, by Theorem
  \ref{thm:pce-complexity}. Considering an instance $\aprob = \phi
  \vdash_{\asid} \psi$ of the safe class, Lemma
  \ref{cor:safe-equiv} reduces this to checking the validity of
  $\card{\decor{\phi}}$ instances of the form $\phi' \vdash_{\rsid}
  \rightdecor{\psi}$, that are all progressing, connected and
  established, by Lemma \ref{cor:pce}. Since a formula $\phi' \in
  \decor{\phi}$ is obtained by replacing each predicate atom
  $p(x_1,\ldots,x_n)$ of $\phi$ by $p_X(x_1,\ldots,x_n,\mvec{\fvar})$
  and there are at most $2^n$ such predicate atoms, it follows that
  $\card{\decor{\phi}} = 2^{\bigO(\width{\aprob})}$. To obtain
  2EXPTIME-membership of the problem, it is sufficient to show that
  each of the progressing, connected and established instances $\phi'
  \vdash_{\rsid} \rightdecor{\psi}$ can be built in time
  $\size{\aprob} \cdot 2^{\bigO(\width{\aprob} \cdot
    \log\width{\aprob})}$. First, for each $\phi' \in \decor{\phi}$,
  by Definition \ref{def:deco}, we have $\size{\phi'} \leq \size{\phi}
  \cdot (1 + \nfv) \leq \size{\phi} \cdot (1 + \width{\aprob}) =
  \size{\phi} \cdot 2^{\bigO(\log \width{\aprob})}$. By Definition
  \ref{def:right_deco}, we have $\size{\rightdecor{\phi}} \leq
  \size{\phi} \cdot (1 + \nfv) = \size{\phi} \cdot 2^{\bigO(\log
    \width{\aprob})}$. By Definition \ref{def:deco_rule},
  $\decor{\asid}$ can be obtained by enumeration in time that depends
  linearly of
  \[\card{\decor{\asid}} \leq \card{\asid} \cdot 2^\maxv \cdot (n+\nfv+\maxv)^\nfv 
    \leq \card{\asid} \cdot 2^{\width{\aprob} + \width{\aprob} \cdot \log\width{\aprob}} 
    = \size{\aprob} \cdot 2^{\bigO(\width{\aprob})}\]
  This is because the number of intervals $I$ is bounded by $2^\maxv$
  and the number of substitutions $\sigma$ by $(n+\nfv+\maxv)^\nfv$,
  in Definition \ref{def:deco_rule}. By Definition
  \ref{def:well-defined}, checking whether a rule is \welldefined can
  be done in polynomial time in the size of the rule, hence in
  $2^{\bigO(\width{\aprob})}$, so the construction of $\leftsid$ takes
  time $\size{\aprob} \cdot 2^{\bigO(\width{\aprob} \log\width{\aprob})}$. Similarly, by
  Definition \ref{def:deco_rule}, the set $\rightdecor{\asid}$ is
  constructed in time 
  \[\card{\rightdecor{\asid}} \leq \card{\asid} \cdot 2^\maxv \cdot \width{\aprob}^\nfv 
  \leq \card{\asid} \cdot 2^\width{\aprob} \cdot 2^{\width{\aprob} \cdot \log \width{\aprob}} 
  = \size{\aprob} \cdot 2^{\bigO(\width{\aprob})}\]
  Moreover, checking that a rule in $\rightdecor{\asid}$ is connected
  can be done in time polynomial in the size of the rule, hence the
  construction of $\rightsid$ takes time
  $2^{\bigO(\width{\aprob}\log\width{\aprob})}$. Then the entire
  reduction takes time $2^{\bigO(\width{\aprob}\log\width{\aprob})}$,
  which proves the 2EXPTIME upper bound for the safe class of entailments. \qed
\end{proof} 

\section{Conclusion and Future Work}

\label{sect:conc}

Together with the results of
\cite{IosifRogalewiczSimacek13,PZ20,DBLP:conf/lpar/EchenimIP20,EIP21a},
Theorem \ref{thm:safe-complexity} draws a clear and complete picture
concerning the decidability and complexity of the entailment problem
in Separation Logic with inductive definitions.  The room for
improvement in this direction is probably very limited, since Theorem
\ref{thm:safe-complexity} pushes the frontier quite far. Moreover,
virtually any further relaxation of the conditions leads to
undecidability.  

A possible line of future research which could be relevant for
applications would be to consider inductive rules constructing
simultaneously several data structures, which could be useful for
instance to handle predicates comparing two structures, but it is
clear that very strong conditions would be required to ensure
decidability. We are also interested in defining effective,
goal-directed, proof procedures (i.e., sequent or tableaux calculi)
for testing the validity of entailment problems.  Thanks to the
reduction devised in the present paper, it is sufficient to focus on
systems that are progressing, connected and left-established.  We are
also trying to extend the results to entailments with formul{\ae}
involving data with infinite domains, either by considering a theory
of locations (e.g., arithmetic on addresses), or, more realistically,
by considering additional sorts for data.

\bibliographystyle{plainurl}
\bibliography{SL}

\newpage

\appendix

\inAppendix

\end{document}

%% file: reductions.pdf_t
\begin{picture}(0,0)%
\includegraphics{reductions.pdf}%
\end{picture}%
\setlength{\unitlength}{1973sp}%
\begingroup\makeatletter\ifx\SetFigFont\undefined%
\gdef\SetFigFont#1#2#3#4#5{%
  \reset@font\fontsize{#1}{#2pt}%
  \fontfamily{#3}\fontseries{#4}\fontshape{#5}%
  \selectfont}%
\fi\endgroup%
\begin{picture}(8716,4517)(-7,-3968)
\put(7276,-3646){\makebox(0,0)[b]{\smash{{\SetFigFont{6}{7.2}{\rmdefault}{\mddefault}{\updefault}{\color[rgb]{0,0,0}\restricted}%
}}}}
\put(4501,-511){\makebox(0,0)[b]{\smash{{\SetFigFont{6}{7.2}{\rmdefault}{\mddefault}{\updefault}{\color[rgb]{0,0,0}(safe)}%
}}}}
\put(4501,239){\makebox(0,0)[b]{\smash{{\SetFigFont{6}{7.2}{\rmdefault}{\mddefault}{\updefault}{\color[rgb]{0,0,0}progressing}%
}}}}
\put(4501,-16){\makebox(0,0)[b]{\smash{{\SetFigFont{6}{7.2}{\rmdefault}{\mddefault}{\updefault}{\color[rgb]{0,0,0}left \connected}%
}}}}
\put(4501,-271){\makebox(0,0)[b]{\smash{{\SetFigFont{6}{7.2}{\rmdefault}{\mddefault}{\updefault}{\color[rgb]{0,0,0}left \restricted}%
}}}}
\put(4351,-3211){\makebox(0,0)[b]{\smash{{\SetFigFont{8}{9.6}{\rmdefault}{\mddefault}{\updefault}{\color[rgb]{0,0,0}$\size{\aprob} \cdot 2^{\bigO(\width{\aprob}^2)}$}%
}}}}
\put(2626,-1711){\rotatebox{45.0}{\makebox(0,0)[b]{\smash{{\SetFigFont{8}{9.6}{\rmdefault}{\mddefault}{\updefault}{\color[rgb]{0,0,0}$\size{\aprob} \cdot 2^{\bigO(\width{\aprob}\log\width{\aprob})}$}%
}}}}}
\put(6751,-1711){\rotatebox{315.0}{\makebox(0,0)[b]{\smash{{\SetFigFont{8}{9.6}{\rmdefault}{\mddefault}{\updefault}{\color[rgb]{0,0,0}$\supseteq$}%
}}}}}
\put(1351,-3136){\makebox(0,0)[b]{\smash{{\SetFigFont{6}{7.2}{\rmdefault}{\mddefault}{\updefault}{\color[rgb]{0,0,0}progressing}%
}}}}
\put(1351,-3391){\makebox(0,0)[b]{\smash{{\SetFigFont{6}{7.2}{\rmdefault}{\mddefault}{\updefault}{\color[rgb]{0,0,0}connected}%
}}}}
\put(1351,-3646){\makebox(0,0)[b]{\smash{{\SetFigFont{6}{7.2}{\rmdefault}{\mddefault}{\updefault}{\color[rgb]{0,0,0}left established}%
}}}}
\put(7276,-3136){\makebox(0,0)[b]{\smash{{\SetFigFont{6}{7.2}{\rmdefault}{\mddefault}{\updefault}{\color[rgb]{0,0,0}progressing}%
}}}}
\put(7276,-3391){\makebox(0,0)[b]{\smash{{\SetFigFont{6}{7.2}{\rmdefault}{\mddefault}{\updefault}{\color[rgb]{0,0,0}\connected}%
}}}}
\end{picture}%

%% file: expansion.pdf_t
\begin{picture}(0,0)%
\includegraphics{expansion.pdf}%
\end{picture}%
\setlength{\unitlength}{1973sp}%
\begingroup\makeatletter\ifx\SetFigFont\undefined%
\gdef\SetFigFont#1#2#3#4#5{%
  \reset@font\fontsize{#1}{#2pt}%
  \fontfamily{#3}\fontseries{#4}\fontshape{#5}%
  \selectfont}%
\fi\endgroup%
\begin{picture}(9923,2627)(518,-2764)
\put(1201,-1711){\makebox(0,0)[b]{\smash{{\SetFigFont{6}{7.2}{\rmdefault}{\mddefault}{\updefault}{\color[rgb]{0,0,0}$\gamma(\ell)$}%
}}}}
\put(1501,-1261){\makebox(0,0)[b]{\smash{{\SetFigFont{6}{7.2}{\rmdefault}{\mddefault}{\updefault}{\color[rgb]{0,0,0}$\aheap$}%
}}}}
\put(5701,-1261){\makebox(0,0)[b]{\smash{{\SetFigFont{6}{7.2}{\rmdefault}{\mddefault}{\updefault}{\color[rgb]{0,0,0}$\mpart{\aheap'}$}%
}}}}
\put(8101,-586){\makebox(0,0)[b]{\smash{{\SetFigFont{6}{7.2}{\rmdefault}{\mddefault}{\updefault}{\color[rgb]{0,0,0}$\bpart{\aheap'}$}%
}}}}
\put(9526,-1336){\makebox(0,0)[b]{\smash{{\SetFigFont{6}{7.2}{\rmdefault}{\mddefault}{\updefault}{\color[rgb]{0,0,0}$\bot$}%
}}}}
\put(9526,-2236){\makebox(0,0)[b]{\smash{{\SetFigFont{6}{7.2}{\rmdefault}{\mddefault}{\updefault}{\color[rgb]{0,0,0}$\bot$}%
}}}}
\put(5101,-2686){\makebox(0,0)[b]{\smash{{\SetFigFont{6}{7.2}{\rmdefault}{\mddefault}{\updefault}{\color[rgb]{0,0,0}$a_1$}%
}}}}
\put(6151,-2686){\makebox(0,0)[b]{\smash{{\SetFigFont{6}{7.2}{\rmdefault}{\mddefault}{\updefault}{\color[rgb]{0,0,0}$a_\rank$}%
}}}}
\put(6601,-2686){\makebox(0,0)[b]{\smash{{\SetFigFont{6}{7.2}{\rmdefault}{\mddefault}{\updefault}{\color[rgb]{0,0,0}$\astore(\fvar_1)$}%
}}}}
\put(7576,-2686){\makebox(0,0)[b]{\smash{{\SetFigFont{6}{7.2}{\rmdefault}{\mddefault}{\updefault}{\color[rgb]{0,0,0}$\astore(\fvar_\nfv)$}%
}}}}
\put(901,-2686){\makebox(0,0)[b]{\smash{{\SetFigFont{6}{7.2}{\rmdefault}{\mddefault}{\updefault}{\color[rgb]{0,0,0}$\gamma(a_1)$}%
}}}}
\put(2176,-2686){\makebox(0,0)[b]{\smash{{\SetFigFont{6}{7.2}{\rmdefault}{\mddefault}{\updefault}{\color[rgb]{0,0,0}$\gamma(a_\rank)$}%
}}}}
\put(3676,-1936){\makebox(0,0)[b]{\smash{{\SetFigFont{8}{9.6}{\rmdefault}{\mddefault}{\updefault}{\color[rgb]{0,0,0}$\gamma$}%
}}}}
\put(9526,-1636){\makebox(0,0)[b]{\smash{{\SetFigFont{6}{7.2}{\rmdefault}{\mddefault}{\updefault}{\color[rgb]{0,0,0}$\bot$}%
}}}}
\put(5626,-2686){\makebox(0,0)[b]{\smash{{\SetFigFont{6}{7.2}{\rmdefault}{\mddefault}{\updefault}{\color[rgb]{0,0,0}$\ldots$}%
}}}}
\put(7126,-2686){\makebox(0,0)[b]{\smash{{\SetFigFont{6}{7.2}{\rmdefault}{\mddefault}{\updefault}{\color[rgb]{0,0,0}$\ldots$}%
}}}}
\put(1576,-2686){\makebox(0,0)[b]{\smash{{\SetFigFont{6}{7.2}{\rmdefault}{\mddefault}{\updefault}{\color[rgb]{0,0,0}$\ldots$}%
}}}}
\put(8101,-1261){\makebox(0,0)[b]{\smash{{\SetFigFont{6}{7.2}{\rmdefault}{\mddefault}{\updefault}{\color[rgb]{0,0,0}$b_1$}%
}}}}
\put(8101,-2161){\makebox(0,0)[b]{\smash{{\SetFigFont{6}{7.2}{\rmdefault}{\mddefault}{\updefault}{\color[rgb]{0,0,0}$b_\maxv$}%
}}}}
\put(8101,-1636){\makebox(0,0)[b]{\smash{{\SetFigFont{6}{7.2}{\rmdefault}{\mddefault}{\updefault}{\color[rgb]{0,0,0}$\vdots$}%
}}}}
\put(9526,-736){\makebox(0,0)[b]{\smash{{\SetFigFont{6}{7.2}{\rmdefault}{\mddefault}{\updefault}{\color[rgb]{0,0,0}$\bot$}%
}}}}
\put(9526,-1111){\makebox(0,0)[b]{\smash{{\SetFigFont{6}{7.2}{\rmdefault}{\mddefault}{\updefault}{\color[rgb]{0,0,0}$\vdots$}%
}}}}
\put(9526,-2011){\makebox(0,0)[b]{\smash{{\SetFigFont{6}{7.2}{\rmdefault}{\mddefault}{\updefault}{\color[rgb]{0,0,0}$\vdots$}%
}}}}
\put(10426,-1111){\makebox(0,0)[b]{\smash{{\SetFigFont{6}{7.2}{\rmdefault}{\mddefault}{\updefault}{\color[rgb]{0,0,0}$\rank+\nfv+\maxv$}%
}}}}
\put(10426,-1936){\makebox(0,0)[b]{\smash{{\SetFigFont{6}{7.2}{\rmdefault}{\mddefault}{\updefault}{\color[rgb]{0,0,0}$\rank+\nfv+\maxv$}%
}}}}
\put(5401,-1786){\makebox(0,0)[b]{\smash{{\SetFigFont{6}{7.2}{\rmdefault}{\mddefault}{\updefault}{\color[rgb]{0,0,0}$\ell$}%
}}}}
\end{picture}%